\documentclass{IEEEtran}

\usepackage{cite}
\usepackage{caption}
\usepackage{indentfirst}
		
\usepackage{amsmath}
\usepackage{amssymb}
\usepackage{graphicx}
\usepackage{epstopdf}
\graphicspath{{fig/}}
\usepackage{subfigure}
\usepackage{float}
\usepackage{color}
\usepackage{array}
\usepackage{longtable}
\usepackage{rotating}
\usepackage{multirow}
\usepackage{amsfonts,amssymb}
\usepackage{mathrsfs}
\usepackage{mdwlist}
\usepackage{threeparttable}
\usepackage{makecell}
\usepackage{diagbox}

\usepackage{amsthm}
\usepackage[linesnumbered,boxed,ruled,commentsnumbered]{algorithm2e}

\usepackage{mdwlist}
\newtheorem{myTheo}{Theorem}
\newtheorem{myLemma}{Lemma}

\newtheorem{myCor}{Corollary}

\title{Speaker Verification By Partial AUC Optimization With Mahalanobis Distance Metric Learning}
\author{{Zhongxin Bai, Xiao-Lei Zhang, and Jingdong Chen}
\thanks{Z. Bai and X.-L. Zhang are with the Center of Intelligent Acoustics and Immersive Communications (CIAIC) and the School of Marine Science and Technology, Northwestern Polytechnical University, 127 Youyi West Road, Xi'an, Shaanxi 710072, China
(e-mail: zxbai@mail.nwpu.edu.cn, xiaolei.zhang@nwpu.edu.cn). }
\thanks{J. Chen is with CIAIC, Northwestern Polytechnical University, 127 Youyi West Road, Xi'an, Shaanxi 710072, China
(e-mail: jingdongchen@ieee.org).}
\thanks{This work was supported in part by the National Key Research and Development Program of China
        under Grant No. 2018AAA0102200 and in part by National Science Foundation
        of China (NSFC) and Israel Science Foundation (ISF) joint research
        program under Grant No. 61761146001 and by the NSFC Key Program under Grant No. 61831019.}}

\begin{document}
\maketitle

\begin{abstract}

Receiver operating characteristic (ROC) and detection error tradeoff (DET) curves are two widely used evaluation metrics for speaker verification. They are equivalent since the latter can be obtained by transforming the former's true positive y-axis to false negative y-axis and then re-scaling both axes by a probit operator. Real-world speaker verification systems, however, usually work on part of the ROC curve instead of the entire ROC curve given an application. Therefore, we propose in this paper to use the area under part of the ROC curve (pAUC) as a more efficient evaluation metric for speaker verification. A Mahalanobis distance metric learning based back-end is applied to optimize pAUC, where the Mahalanobis distance metric learning guarantees that the optimization objective of the back-end is a convex one so that the global optimum solution is achievable.
 To improve the performance of the state-of-the-art speaker verification systems by the proposed back-end, we further propose two feature preprocessing techniques based on  length-normalization  and  probabilistic linear discriminant analysis respectively. We evaluate the proposed systems on the major languages of NIST SRE16 and the core tasks of SITW. Experimental results show that the proposed back-end outperforms the state-of-the-art speaker verification back-ends  in terms of seven evaluation metrics.
\end{abstract}


\section{Introduction}
\IEEEPARstart{S}{peaker} verification aims to verify whether an utterance is pronounced by a hypothesized speaker based on some utterances pre-recorded from that speaker. Depending on whether it requires the to-be-verified speaker to pronounce some pre-defined text or not, speaker verification can be classified into two classes, i.e., \textit{text-dependent} and \textit{text-independent}. This paper focuses on the text-independent case. There are generally two approaches to this problem: a two-step one, which consists of a front-end feature extractor and a back-end classifier, and a one-step approach, which trains an end-to-end system \cite{heigold2016end,zhang2016end,snyder2016deep,jung2018complete}. This paper focuses on the two-step approach.

In a two-step approach, it is important to have a good front-end. In the literature, the Gaussian mixture model (GMM) based universal background model (UBM) \cite{reynolds2000speaker} plus identity vector (i-vector) \cite{dehak2011front} is commonly used. In such a front-end, a GMM-UBM is first trained to collect Baum-Welch statistics, which is formed as a supervector for each utterance. Then, factor analysis is used to reduce the dimensionality of the supervectors to low-dimensional i-vectors. Many extensions of the GMM-UBM/ivector front-end were proposed recently, e.g., \cite{cumani2018speaker}. Motivated by the paradigm shift of speech recognition from GMM-based acoustic modeling to deep neural network (DNN) based one, a DNN-UBM/i-vector front-end was developed \cite{lei2014novel,kenny2014deep,richardson2015deep}. It essentially uses the DNN-based acoustic model trained for speech recognition to generate the posterior probabilities instead of GMM-UBM. Tan \textit{et al.} further employed a denoising autoencoder to replace the DNN-based acoustic model for dealing with environmental noise \cite{tan2018denoised}. These method, however, needs transcriptions of the training data to train the acoustic models, which may not be always available.

An emerging direction of the front-end research is deep embedding. Deep embedding uses a DNN to distinguish the training speakers in a closed set by a classification-based loss function, and takes the outputs of the hidden layers of the DNN for verification.
An early deep embedding front-end is the so-called d-vector \cite{variani2014deep,li2017deep}, in which frame-level speaker features are extracted from the top
hidden layer, and then utterance-level speaker features are derived as the average of the frame-level features.
However, the average of the frame level features does not consider the dependency of the contextual frames. Several efforts have been made to address this problem \cite{snyder2017deep,snyder2018xvector,okabe2018attentive,zhu2018self}.
For example, in \cite{snyder2017deep,snyder2018xvector}, Snyder \textit{et al.} proposed to insert an average pooling layer into DNN to handle variable-length segments. In \cite{gao2018improved}, Gao \textit{et al.} exploited a cross-convolutional-layer pooling method to extract the first-order statistics of the input segments. Attention mechanism was also studied to generate utterance-level features \cite{okabe2018attentive,zhu2018self}. Another problem with the deep embedding front-end is on the training loss function. Because the classification-based loss is only a surrogate loss function of the final evaluation metrics of speaker verification, finding more effective loss functions become an important issue. In \cite{yadav2018learning,li2018deep}, the authors proposed to minimize the classification-based loss and center loss together. In \cite{zhang2018text},  Zhang \textit{et al.} took triplet loss  as the training objective of a deep embedding network. Although employing the above training loss functions is shown to be able to improve the performance, the extracted speaker features still have significant intra-class variations, which need to be handled by back-ends.

Regarding the back-end, commonly used back-end classifiers include cosine similarity scoring \cite{dehak2011front}, support vector machine \cite{cumani2014large}, and probabilistic linear discriminant analysis (PLDA) \cite{ioffe2006probabilistic,kenny2010bayesian,garcia2011analysis}. DNNs have also been investigated \cite{ghahabi2014deep,ghahabi2017deep}.
Inter-session variability compensation is a main task of back-ends, since the front-ends are inter-session- and speaker-dependent. Linear compensation techniques such as linear discriminant analysis (LDA) and within class covariance normalization \cite{hatch2006within} are often used. Recently, nonlinear compensation methods have been studied as well: Cumani \textit{et al.}\cite{cumani2017nonlinear,cumani2017joint} proposed a nonlinear transformation to i-vectors to make them more suitable for PLDA \cite{cumani2018scoring}; Zheng \textit{et al.} developed a DNN-based dimensionality reduction method as an alternative to LDA \cite{tieran2018deep}.
However, because the aforementioned back-ends do not optimize the evaluation metrics directly, such as equal error rate (EER), their performance  may be suboptimal.

To optimize the evaluation metrics directly, metric learning needs to be used, which attempts to learn an appropriate similarity measurement space of data points. It has been widely studied in the machine learning community. One of the most popular metric learning methods is to optimize the parameters of a Mahalanobis distance in a linear space \cite{kulis2013metric}. Recently, deep metric learning \cite{oh2016deep,hoffer2015deep,cakir2019deep}, which uses a DNN to learn a nonlinear similarity measurement, has also received much attention. Metric learning has been recently studied in speaker verification as well. For example, some metric learning based back-ends \cite{bai2018cosine,novoselov2018triplet} have been proposed to compensate the inter-session variability of the embedding features, where the work in \cite{bai2018cosine} minimizes the EER of speaker verification directly. It is also popular to train an end-to-end speaker verification system \cite{heigold2016end,zhang2016end,snyder2016deep,jung2018complete} or an embedding DNN  \cite{zhang2018text} by deep metric learning. In our recent work \cite{bai2018cosine}, we proposed a linear cosine metric learning algorithm to minimize the overlap region of decision scores. Similarly, in \cite{novoselov2018triplet}, Novoselov \textit{et al.} proposed a triplet-loss-based cosine similarity metric learning back-end.

Although directly optimizing an evaluation metric of speaker verification improves the performance,  current methods focus mainly on optimizing EER. Since it needs to work at a different point of its receiver operating characteristic (ROC) curve for different applications, a speaker verification system tuned to yield the minimum EER in one scenario may not produce the best performance in another scenario.
To address this issue, this paper proposes a back-end to directly optimize part of the area under the ROC curve (named \textit{partial AUC}, or pAUC for short). The main contributions of this paper are summarized as follows:
\begin{itemize}

\item A new calibration-insensitive evaluation metric named ``pAUC'' is proposed for speaker verification. pAUC represents partial area under the ROC curve. It meets the evaluation requirement of real-world applications that work on different parts of ROC curves, such as bank security systems or terrorist detection systems. It is a supplement evaluation metric to the existing metrics. As shown in Fig. \ref{fig:pAUC_schematic}, the pAUC for a specific application is defined by two false positive rate (FPR) parameters: $\alpha$ and $\beta$.

 \item A Mahalanobis metric learning back-end is proposed to maximize pAUC (pAUCMetric). pAUCMetric evaluates the similarity between two speaker features by a squared  Mahalanobis distance, and optimizes the parameters of the distance metric to maximize pAUC where the working points of the speaker verification system locate. pAUCMetric is formulated as a convex optimization problem, where the global optimum solution is guaranteed.  We further combine pAUCMetric with two feature preprocessing techniques: 1) length-normalization, and 2) latent variables of PLDA, which combine the ranking property of pAUC into the Cosine similarity or PLDA back-ends for further performance improvement. It is shown that the AUC optimization, such as the one in \cite{garcia2012optimization,mingote2019optimization},  can be viewed as a special case of pAUC with $\alpha=0$ and $\beta=1$.

\end{itemize}

Experiments are conducted to evaluate the effectiveness of pAUCMetric and compare pAUCMetric with PLDA and cosine similarity scoring back-ends that do not optimize evaluation metrics directly. For each experiment, all back-ends use the same front-end, which is either the GMM/i-vector or the x-vector. We train the comparison methods on switchboard, NIST SRE04--SRE10 and VoxCeleb datasets, and evaluate them on the major languages of NIST SRE16 and the core tasks of SITW. The evaluation is conducted under the conditions of both noise-matching and -mismatching, as well as both language-matching and -mismatching.
The experimental results show that pAUCMetric outperforms PLDA by relatively $10\%$, $9\%$ and $20\%$ in terms of EER, pAUC and AUC metrics respectively.

The rest of this paper is organized as follows:  Section \ref{sec:mot_rel} presents the motivations. Section \ref{sec:proposed_algorithm} and \ref{sec:feature} describe the proposed algorithm.  Section \ref{sce:experiment} presents the experiment results. Finally, important conclusions are drawn in section \ref{sec:summary}.

\section{Motivation} \label{sec:mot_rel}

\begin{figure}[t!]
  \centering
  \includegraphics[width=2.2 in]{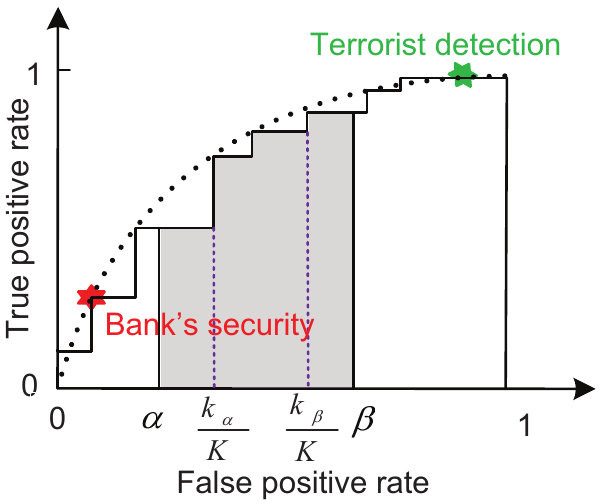}
  \caption{Illustrations of the ROC curve, AUC, and pAUC. }
  \label{fig:pAUC_schematic}
\end{figure}

\subsection{Motivation for the pAUC evaluation metric}

\begin{figure}[t!]
  \centering
  \includegraphics[width=3.3 in]{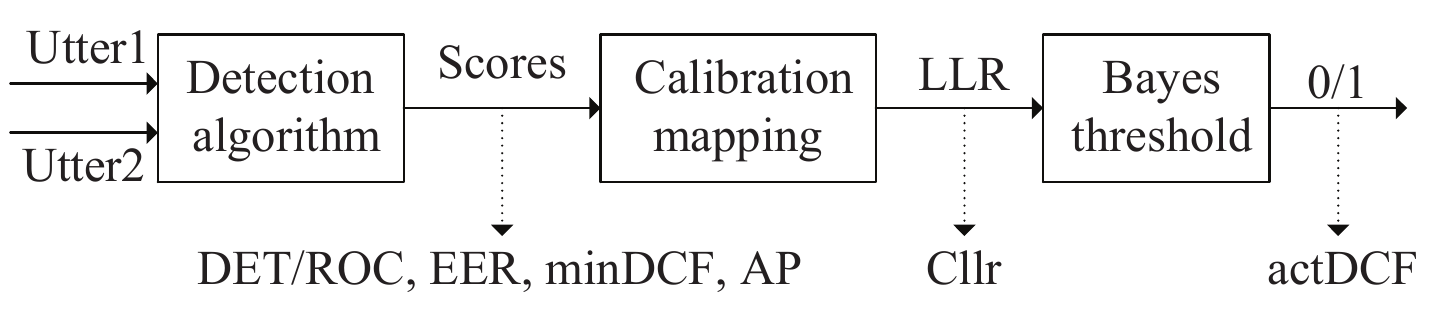}
  \caption{Diagram of a speaker verification system with common evaluation metrics. }
  \label{fig:SV_diagram}
\end{figure}

It is known that a speaker verification system first generates a similarity score of a trial by a speaker detection algorithm, and then makes a hard decision according to a threshold as illustrated in Fig.~\ref{fig:SV_diagram}. The speaker detection algorithm assigns higher scores to \textit{target trials} than \textit{non-target trials}, which determines the \textit{discriminability} of the system.
The decision threshold is usually determined by first \textit{calibrating} the similarity scores to the log-likelihood ratios (LLR) and then applying the Bayes decision theory \cite{brummer2013bosaris} using application-dependent priors, i.e., the prior of targets  and the costs of false negative rate (also known as miss detection rate) and false positive rate (also known as false alarm rate).

The evaluation metrics of speaker verification in Fig.~\ref{fig:SV_diagram} can be categorized to two classes---\textit{calibration-sensitive} metrics and \textit{calibration-insensitive} ones. The calibration-sensitive metrics, which include the actual detection cost function (actDCF) and cost of LLR ($\rm C_{llr}$), aim to evaluate a calibrated speaker verification system under the framework of Bayes decision theory. Specifically, the application-dependent actDCF evaluates the empirical Bayes risk of a system at the  Bayes decision threshold \cite{brummer2013bosaris}, which determines how good is the hard decision. $\rm C_{llr}$ evaluates the discrimination of the calibrated LLR in an application-independent manner \cite{brummer2006application}. While calibration-sensitive evaluation metrics have many pros in evaluating the suitability of a calibrated system, we often need to evaluate the detection algorithm 
of an uncalibrated system directly.

In contrast, calibration-insensitive metrics evaluate the discriminability of  the detection algorithm. They include the detection error tradeoff (DET) curve, EER, minimum detection cost function (minDCF), and average precision. DET curve is an alternative form of the ROC curve. As a matter of fact, the DET curve can be obtained by transforming the ROC curve's true positive y-axis to false negative y-axis and then re-scaling both axes by a non-linear warping named the probit operator\cite{brummer2013bosaris, martin1997det}. It reflects the global discriminability of a speaker verification system.
EER and minDCF are two points on the DET curve, which reflect the discriminability of the system to some extent. Like the DET curve, average precision is a global metric that combines recall and precision for ranked retrieval results, which is however  sensitive to class-imbalanced problems such as speaker verification. To summarize, DET curve and average precision are two global metrics, while EER and minDCF are two local points on the DET curve.

In practice, a speaker verification system usually works on a local fraction of the DET curve with a tunable threshold, instead of a single local point. For example, a bank security system is tuned in a range where the false positive rate is controlled below $0.01\%$. In contrast, a terrorist detection system of a public security
department is tuned in a range whose recall rate is required in a range of higher than 99\%. As shown in Fig. \ref{fig:pAUC_schematic}, pAUC may meet such a requirement. First, $[\alpha,\beta]$ in Fig. \ref{fig:pAUC_schematic} defines the interested operating points of a real-world working scenario. Second, pAUC, which is a scalar in the range of $[0,1]$, describes the interested part of the ROC curve efficiently. At last, its calculation method, which will be presented in \eqref{eq:EmpAUC}, does not depend on a decision threshold. Hence, we adopt pAUC as a new calibration-insensitive evaluation metric.

\subsection{Motivation for the pAUCMetric back-end}

How to optimize calibration-sensitive evaluation metrics has been well studied and a number of methods were developed \cite{brummer2010measuring,brummer2013likelihood}. But those methods do not improve the discriminability of the detection algorithm as the order of the similarity scores of training trials is not changed. In order to improve the discriminability of the detection algorithm, it is better to optimize the ROC curve directly by maximizing its AUC. However, optimizing the entire AUC is not only costly but also unnecessary as that most practical systems work only on part of their ROC curves. Therefore, we propose a metric learning back-end based on Mahalanobis distance to optimize pAUC accordingly.

Another advantage of pAUCMetric is that it can select difficult negative training trials by setting $\beta$ to a small value, which is a well-known challenging problem for the algorithms that need to group training utterances into training trials.
As will be shown in the experiments, the proposed pAUCMetric performs better than a triplet-loss-based algorithm, which differs from pAUCMetric only in the loss function, for all the aforementioned evaluation metrics.

\section{pAUC metric learning back-end} \label{sec:proposed_algorithm}
 In this section, we first provide an overview to the speaker verification system in Section \ref{sec:euc}, and then present the objective function  and optimization algorithm of the proposed back-end in Sections \ref{sec:obj} and \ref{subsec:opt_alg} respectively.

\subsection{System overview}\label{sec:euc}

\begin{figure}[t!]
  \centering
  \includegraphics[width=2.8 in]{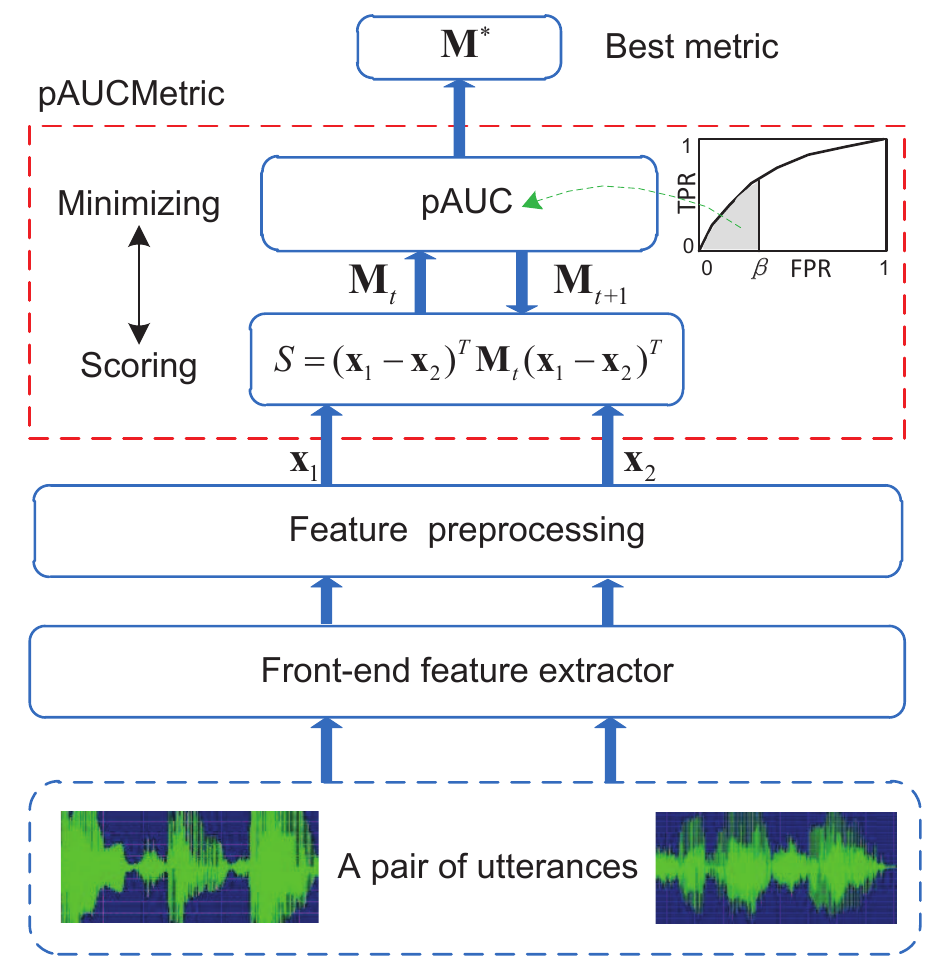}
  \caption{Diagram of the  pAUCMetric based speaker verification system. }
  \label{fig:back_end_framework}
\end{figure}
The diagram of the pAUCMetric based speaker verification system is shown in Fig. \ref{fig:back_end_framework}. The front-end is used to extract speaker features from speech signals. We use i-vector\cite{dehak2011front} or x-vector\cite{snyder2018xvector} as the front-end. After feature extraction by the front-end, we further preprocess the features as described in Section \ref{sec:feature}, and then use the preprocessed feature as the input of pAUCMetric.

The role of pAUCMetric is to judge whether two preprocessed features  $\mathbf{x}_{q_1}$ and $\mathbf{x}_{q_2}$ belong to the same speaker based on their similarity. The similarity is measured by the following squared Mahalanobis distance:
\begin{align}
\label{eq:Mahalanobis}
    S(\mathbf{x}_{q_1},\mathbf{x}_{q_2};\mathbf{M}) =(\mathbf{x}_{q_1}-\mathbf{x}_{q_2})^T\mathbf{M}(\mathbf{x}_{q_1}-\mathbf{x}_{q_2})
\end{align}
where $\mathbf{M}$ is a symmetric positive semi-definite matrix, which is to be learned by pAUCMetric. If the squared Mahalanobis distance between  $\mathbf{x}_{q_1}$ and $\mathbf{x}_{q_2}$  is smaller than a pre-specified threshold $\theta^*$, $\mathbf{x}_{q_1}$ and $\mathbf{x}_{q_2}$ are regarded as from the same speaker; otherwise, they are regarded as from different speakers. We denote $\mathbf{z}_i=\mathbf{x}_{q_1}-\mathbf{x}_{q_2}$, and denote $ S(\mathbf{x}_{q_1},\mathbf{x}_{q_2};\mathbf{M})$ as $S(\mathbf{z}_i;\mathbf{M})$ for simplicity. A probabilistic explanation of  the Mahalanobis distance is given in Appendix A.

\subsection{Objective function}
\label{sec:obj}

Given a training set with $N$ speakers and $Q$ embedding vectors  $\mathcal{X}=\{(\mathbf{x}_q,y_q)\}_{q=1}^Q$, where $y_q = 1,\ldots,N$ is the identity of $\mathbf{x}_q$,
we first construct a pairwise training set
 \begin{equation}\label{eq:T}
  \mathcal{T}=\{(\mathbf{z}_i,l_i)\}_{i=1}^I
\end{equation}
where $\mathbf{z}_i=\mathbf{x}_{q_1}-\mathbf{x}_{q_2}$ with $q_1=1,\ldots,Q$ and $q_2=1,\ldots,Q$ ($q_1\neq q_2$), $I$ is the size of $\mathcal{T}$, and $l_i$ is the ground-truth label of $\mathbf{z}_i$ satisfying:
  \begin{equation}\label{eq:li}
  l_i=\left\{\begin{array}{ll}
    1,&\quad \mbox{if } y_{q_1}=y_{q_2}\\
    -1,& \quad \mbox{otherwise}
  \end{array}\right.
\end{equation}
 We define the subset of the true trials of $\mathcal{T}$ as:
\begin{equation}\label{eq:true}
  \mathcal{P}=\{(\mathbf{z}_j^+,l_j=1) \}_{j=1}^J
\end{equation}
and the subset of the imposter trials of $\mathcal{T}$ as:
\begin{equation}\label{eq:imposter}
 \mathcal{N}=\{ (\mathbf{z}_k^-,l_k=-1)\}_{k=1}^K
\end{equation}
where $J$ and $K$ are the sizes of $\mathcal{P}$ and $\mathcal{N}$ respectively.


After the above preliminary setting, the pAUC is calculated as follows. We define a subset of $\mathcal{N}$ that defines the pAUC over the FPR range $[\alpha, \beta]$:
\begin{equation}\label{eq:N0}
 \mathcal{N}_0=\{ (\mathbf{z}_r^-,l_r=-1)\}_{r=1}^R
\end{equation}
where $R\leq K$, and $\mathcal{N}_0$ is determined as following.
 Because the imposter set $\mathcal{N}$ contains only a limited number of trials, we first replace $[\alpha$,$\beta]$ by $\left[k_\alpha/K,k_\beta/K \right]$ where $k_\alpha=\left \lceil K \alpha \right \rceil$  and $k_\beta=\left \lfloor K \beta \right \rfloor $ are two integers. Then,  $\{S(\mathbf{z}_k^-;\mathbf{M})\}_{\mathbf{z}_k^-\in \mathcal{N}}$ are sorted in ascending order. Finally,  we pick the trials ranked from the top $k_\alpha$th to $k_\beta$th positions to form $\mathcal{N}_0$. The calculation of pAUC is equivalent to that of the normalized AUC over $\mathcal{P}$ and $\mathcal{N}_{0}$, which is computed as:
   \begin{align}\label{eq:EmpAUC}
     {\rm pAUC} =& 1-\frac{1}{JR} \sum_{j=1}^{J} \sum_{r=1}^{R}
     \left[ \mathbb{I}(S(\mathbf{z}_j^+;\mathbf{M})> S(\mathbf{z}_r^-;\mathbf{M})) \right .   \nonumber \\
      &+ \frac{1}{2}\mathbb{I}(S(\mathbf{z}_j^+;\mathbf{M})= S(\mathbf{z}_r^-;\mathbf{M})) \left . \right]
   \end{align}
where  $\mathbb{I}(\cdot)$ is an indicator function that returns 1 if the statement is true, and 0 otherwise.

However, directly optimizing \eqref{eq:EmpAUC} is an NP-hard problem. To circumvent this, let us relax \eqref{eq:EmpAUC} by replacing the indicator function by a hinge loss function:
   \begin{eqnarray}\label{equal:hinge_loss}
      &\ell_{\mbox{hinge}}(S(\mathbf{z}_j^+;\mathbf{M})>S(\mathbf{z}_r^-;\mathbf{M}))=\nonumber\\
      &{\rm max} \left[0,\delta-\left(S(\mathbf{z}_r^-;\mathbf{M})-S(\mathbf{z}_j^+;\mathbf{M})\right) \right]
   \end{eqnarray}
where $\delta>0 $ is a tunable hyper-parameter controlling  the distance margin between $\{S(\mathbf{z}_r^-;\mathbf{M})\}_{\mathbf{z}_r^-\in \mathcal{N}_0}$ and $\{S(\mathbf{z}_j^+;\mathbf{M})\}_{\mathbf{z}_j^+\in \mathcal{P}}$. Substituting \eqref{equal:hinge_loss} into \eqref{eq:EmpAUC} and further changing the maximization problem \eqref{eq:EmpAUC} into an equivalent minimization one gives \eqref{eq:loss}.
  \begin{equation}
  \label{eq:loss}
     \ell=\frac{1}{JR}\sum_{j=1}^{J}\sum_{r =1}^{R} {\rm max} \Big( 0,\delta-S(\mathbf{z}_r^-;\mathbf{M})+S(\mathbf{z}_j^+;\mathbf{M})\Big)
   \end{equation}

The proposed pAUCMetric minimizes \eqref{eq:loss} over $\mathcal{P}$ and $\mathcal{N}_0$.
To prevent overfitting to the training data, we add a regularization term $\lambda\Omega(\cdot)$ to the minimization problem according to a plausible formulation in \cite{huo2018cross}, which gives the objective function of pAUCMetric:
\begin{equation}\label{eq:structure_loss}
   \mathbf{M}^*=\arg \min_{\mathbf{M}}\ell(\mathcal{P},\mathcal{N}; \mathbf{M} )+\lambda\Omega(\mathbf{M}),
\end{equation}
where $\lambda$ is a regularization hyperparameter, and $\lambda\Omega(\cdot)$ is defined as:
\begin{equation}\label{eq:regular}
   \lambda \Omega(\mathbf{M})=\frac{\gamma}{J}\sum_{j=1}^{J}S(\mathbf{z}_j^+;\mathbf{M})+\mu[{\rm tr}(\mathbf{M})-{\rm logdet}(\mathbf{M})]
\end{equation}
with $\gamma$ and $\mu$ being two tunable hyper-parameters. The first term on the right-hand side of (\ref{eq:regular}), i.e., $\frac{1}{J}\sum_{j=1}^{J}S(\mathbf{z}_j^+;\mathbf{M})$, which was first introduced in \cite{weinberger2009distance}, aims to bound $S(\mathbf{z}_j^+;\mathbf{M})$ in \eqref{eq:loss}.  The second term, i.e., ${\rm tr}(\mathbf{M})-{\rm logdet}(\mathbf{M})$, which is a specifical case of \textit{LogDet divergence} \cite{davis2007information} defined over positive semi-definite (PSD) matrices \cite{kulis2013metric}, is used to improve the generalization ability and further constrain $\mathbf{M}$ to be PSD.

\eqref{eq:structure_loss} can also be interpreted from another viewpoint using the following lemma.
\begin{myLemma}\label{cor:1}
  The maximization of pAUC in \eqref{eq:structure_loss} is a problem of enlarging a weighted margin between the positive and negative trials while minimizing the within-class variances of the two class trials simultaneously.
\end{myLemma}
\begin{proof}
Let us define an index matrix $\mathbf{\Pi}\in \{0, 1\}^{J \times R} $:
\begin{equation}\label{index_matrix}
   \mathbf{\Pi}(j,r)=\left \{
   \begin{array}{cl}
          1, &  \quad\mbox{if    } \delta+S(\mathbf{z}_j^+;\mathbf{M})  >   S(\mathbf{z}_r^-;\mathbf{M}) \\
          0, & \quad\mbox{otherwise} \\
    \end{array}.
    \right.
\end{equation}
and rewrite the loss function of \eqref{eq:loss} as:
\begin{align}\label{eq:loss_probxxx}
   \ell  & = \frac{\delta }{JR}\sum_{j =1}^{J}\sum_{r = 1}^{R} \mathbf{\Pi}(j,r) +\frac{1}{J}\sum_{j =1}^{J}\Big(\frac{1}{R}\sum_{r=1}^R \mathbf{\Pi}(j,r)\Big)S(\mathbf{z}_j^+;\mathbf{M})  \nonumber \\
     & \quad -\frac{1}{R}\sum_{r=1}^{R}\Big(\frac{1}{J}\sum_{j=1}^{J}\mathbf{\Pi}(j,r)\Big)S(\mathbf{z}_r^-;\mathbf{M})  \nonumber \\
     &=c+\frac{1}{J}\sum_{j=1}^{J}p_j S(\mathbf{z}_j^+;\mathbf{M}) -\frac{1}{R}\sum_{r =1}^{R}p_r S(\mathbf{z}_r^-;\mathbf{M})
\end{align}
where $c= \frac{\delta }{JR}\sum_{j =1}^{J}\sum_{r = 1}^{R} \mathbf{\Pi}(j,r)$ is a constant in a single iteration, $p_j=\frac{1}{R}\sum_{r=1}^{R}\mathbf{\Pi}(j,r)$ and  $p_r=\frac{1}{J}\sum_{j=1}^{J}\mathbf{\Pi}(j,r)$ are the weights of the positive and negative trials respectively.
It is clear that minimizing \eqref{eq:loss_probxxx} is a problem of enlarging the weighted margin between the positive and negative trials.
\end{proof}

Because the regularization term $\frac{\gamma}{J}\sum_{j=1}^{J}S(\mathbf{z}_j^+;\mathbf{M})$ minimizes the within-class variance, we see that the objective \eqref{eq:structure_loss} enlarges the between-class distance and minimizes the within-class variance simultaneously, which is also the principle behind many well-known back-ends, such as LDA, WCCN, and PLDA. The difference lies in that pAUCMetric works in the squared Mahalanobis distance space and encodes the pAUC information into the weights $p_j$ and $p_r$.

\subsection{Optimization algorithm}\label{subsec:opt_alg}

In order to solve the optimization problem in (6), substituting \eqref{index_matrix} into  \eqref{eq:structure_loss} gives
\begin{equation}\label{eq:loss_prob}
       \mathbf{M}^*=\arg\min_{\mathbf{M}}\langle \mathbf{P}+ \gamma  \mathbf{P}_{\mathcal{P}}, \mathbf{M}\rangle _F+\mu \left[{\rm tr}(\mathbf{M})-{\rm logdet}(\mathbf{M})\right],
\end{equation}
where $\langle \cdot \rangle_F$ denotes the Frobenius norm operator, and
\begin{eqnarray}\label{index_matrix2}
  &&\mathbf{P}_{\mathcal{P}}= \frac{1}{J} \sum_{j=1}^{J}\mathbf{z}_j^+\mathbf{z}_j^{+T}, \\\label{index_matrix3}
  &&\mathbf{P}= \frac{1}{JR}\sum_{j=1}^{J}  \sum_{r=1}^{R}\mathbf{\Pi}(j,r)(\mathbf{z}_j^+\mathbf{z}_j^{+T}-\mathbf{z}_r^-\mathbf{z}_r^{-T}).
\end{eqnarray}
We employ the {proximal point algorithm} (PPA)\cite{huo2018cross}  to optimize \eqref{eq:loss_prob}. The resulting algorithm, which is summarized in Algorithm 1, consists of the following three steps at each iteration:

\begin{itemize}
\item The first step constructs the training set $\mathcal{T}$ from $\mathcal{X}$. However, if we consider all trials in $\mathcal{X}$ during the construction of $\mathcal{T}$, the size of $\mathcal{T}$ becomes enormous. To prevent the overload of computing, we construct a pairwise set $\mathcal{T}^t$ at each iteration by a random sampling strategy as follows. We first randomly select $s$ speakers from $\mathcal{X}$, then randomly select two embedding vectors from each of the selected speakers, and finally construct $\mathcal{T}^t$ by a full permutation of the $2s$ embedding vectors. $\mathcal{T}^t$ contains $s$ true training trials and $s(2s-1)-s$ imposter training trials.

\item The second step calculates ${\mathcal{N}}_{0}^t$ according to \eqref{eq:N0},  and calculates $\mathbf{P}^t$ and $\mathbf{P}_\mathcal{P}^t$ according to \eqref{index_matrix2} and \eqref{index_matrix3} respectively.

\item The third step updates $\mathbf{M}$ by PPA \cite{huo2018cross}, which first applies eigenvalue decomposition to $\mathbf{X}=\mathbf{M}_t-\eta(\mathbf{P}^t+\gamma \mathbf{P}_\mathcal{P}^t+\mu \mathbf{I}_0) $, i.e., $\mathbf{X}=\mathbf{UVU}^T$ where $\mathbf{V}={\rm diag}([v_1,v_2,\cdots,v_d])$ with $v_1 \geq v_2 \geq \cdots \geq v_d$, and then adopts the following updating equation:
\begin{equation}
   \phi_{\lambda}^+(\mathbf{x})=\mathbf{U}{\rm diag} ([\phi_{\lambda}^+(v_1),\cdots,\phi_{\lambda}^+(v_d)])\mathbf{U}^T,
\end{equation}
where $\phi_{\lambda}^+(v)=\left[(v^2+4\lambda)^{1/2}+v\right]/2$, and $d$ is the  dimension of the input feature.
\end{itemize}

\subsection{Connection to the back-ends trained with training trials.}\label{sec:trip}

There are two basic classes of back-ends depending on how they construct the training data. One class takes training utterances as the training data for training a generative PLDA. The other groups training utterances into training trials for training binary-class classifiers, in which back-ends differ in two aspects---basic classifiers and loss functions. Here we focus on discussing the difference between the loss functions of the second class, which include the pairwise SVM \cite{burget2011discriminatively,cumani2011fast,cumani2014large}, triplet-loss-based, and pAUCMetric back-ends whose loss functions are denoted as the classification-loss, triplet-loss, and pAUC-loss \eqref{eq:loss}, respectively.

The classification-loss \cite{cumani2014large}, triplet-loss \cite{zhang2018text}, and pAUC-loss all use the hinge loss function to relax the $0/1$-loss. The only difference between them is how the errors are accumulated.
The classification-loss accumulates the classification error, which suffers from the class-imbalance problem of speaker verification. In contrast, the pAUC-loss focuses on the ranking of the similarity scores; So, it does not suffer from the class-imbalance problem.

The triplet-loss requires that the features from the same speaker are closer than those from different speakers in a triplet trial  \cite{zhang2018text}, i.e.,
\begin{equation}
\label{eq:trip}
    S(\mathbf{x}^a,\mathbf{x}^n;\mathbf{M})-S(\mathbf{x}^a,\mathbf{x}^p;\mathbf{M}) >  \delta
\end{equation}
where $\mathbf{x}^a$, $\mathbf{x}^p$, and $\mathbf{x}^n$ represent, respectively, the anchor, positive, and negative utterances of a trial. For clarity, we denote the speaker features in a constraint as a \textit{relative constraint}. For example, we call $\{\mathbf{x}^a, \mathbf{x}^p, \mathbf{x}^n\}$ in \eqref{eq:trip} as a {relative constraint} of the triplet-loss. The difference between the triplet-loss \eqref{eq:trip} and pAUC-loss \eqref{eq:loss} lies in the following three aspects.

First, the relative constraints of the triplet-loss \eqref{eq:trip} are triplet, which cannot deal with the situation where the training data contains only positive or negative trials. While the relative constraints of the pAUC-loss \eqref{eq:loss} are tetrad, which matches the pipeline of speaker verification. Therefore, \eqref{eq:loss} does not have the same limitation as \eqref{eq:trip}. Second, the pAUC-loss is intrinsically able to pick difficult training trials from the exponentially large number of training trials, while the triplet-loss lacks such an ability (that is why it has to use additional training trial selection methods). Third, as proven in Appendix B, the relative constraints of the triplet-loss are a subset of the relative constraints of the AUC-loss. As a result, the triplet-loss is a specifical case of the AUC-loss. Moreover, because the AUC-loss is a specifical case of the pAUC-loss with $\alpha= 0$ and $\beta=1$, we conclude that the triplet-loss is a specifical case of the pAUC-loss.

To validate the above analysis, we propose a new triplet-loss based algorithm for experimental comparison, which differs from pAUCMetric only in the loss function. The new algorithm, named TripletMetric, replaces the tetrad constraints \eqref{eq:loss} with the triplet constraints \eqref{eq:trip}. Its training data are constructed in a similar way as $\mathcal{T}$ in Section \ref{subsec:opt_alg}, which randomly selects $s$ speakers with each speaker selecting two embedding vectors. The number of the training triplet trials is $2s(2s-2)$. Note that, because the number of the tetrad constraints in \eqref{eq:loss} is $s[(s(2s-1)-s)(\beta-\alpha)]$, the ratio of the number of the training trials of pAUCMetric to that of TripletMetric is $\frac{2}{s(\beta-\alpha)}$.

\begin{algorithm}[!t]
\label{alg:PPA}
  \SetKwInOut{Input}{\textbf{Require}}   
  \SetKwInOut{Output}{\textbf{Output}}
  \caption{ Mini-batch PPA \cite{huo2018cross} algorithm for pAUCMetric  optimization.}
  \Input{\\
         Development set: $\mathcal{X}$;    \\
         False positive rate: $\alpha \geq 0,\beta>0$; \\
         Hyperparameter: $\delta \geq  0$,  $\gamma \geq 0$, $ \mu \geq  0$; \\
         Batch size: $s$; \\
         Step size parameter: $\eta>0$; \\
         Initialize: $ t \leftarrow 0 $, $\mathbf{M}_0=\mathbf{I}_0$, where $\mathbf{I}_0$ is the\\
         ~~~~~~~~~~~~identity matrix;
        }
    \BlankLine
    \Repeat{\text{converged}}
        {
          Construct a mini-batch subset of $\mathcal{X}$ by random sampling; \\
          Construct $\mathcal{T}$ from the subset of $\mathcal{X}$ by \eqref{eq:T}\\
          Compute $\mathcal{P}$ and $\mathcal{N}_0$ by \eqref{eq:true} and \eqref{eq:N0}; \\
          Calculate $\mathbf{P}^t$ and $\mathbf{P}_\mathcal{P}^t$ on  $\mathcal{P}$ and ${\mathcal{N}_0}$ by \eqref{index_matrix2} and \eqref{index_matrix3}; \\
          $\mathbf M_{t+1}\leftarrow \phi_{\lambda}^+(\mathbf{M}_t-\eta(\mathbf{P}^t+\gamma \mathbf{P}_\mathcal{P}^t+\mu \mathbf{I}_0))$, where $\lambda=\eta\mu$;\\
          $ t \leftarrow t+1 $. \\
          }
    \Output{$\mathbf{M}_t$ }
\end{algorithm}

\section{Complexity analysis}
\label{sec:complexity}


\begin{myTheo}\label{thm:1}
 The computational complexity  of pAUCMetric is:

\begin{equation}\label{eq:complexity}
    O= \mathcal{O}(d^2(I+J+R))+\mathcal{O}(JR)+\mathcal{O}(d^3)+\mathcal{O}(K{log}_2K)
\end{equation}
where $I$, $J$, $R$, and $K$ are the size of $\mathcal{T}$, $\mathcal{P}$, $\mathcal{N}_0$, and $\mathcal{N}$, respectively, and $d$ is the dimension of the input feature. 
\end{myTheo}

\begin{proof}
According to Algorithm 1, the computational complexity of pAUCMetric is composed of three parts:

The first part is the computation of  $\mathcal{P}$ and $\mathcal{N}_0$.  We first need $\mathcal{O}(I)$ operations to separate the positive and negative trials in $\mathcal{T}$. Then, computing the squared  Mahalanobis distances between all training pairs according to \eqref{eq:Mahalanobis} consumes $\mathcal{O}(d^2I)$ multiplications. Finally, we need $\mathcal{O}(K{log}_2K)$ operations to sort all scores of $\mathcal{N}$ for $\mathcal{N}_0$. Thus, the total computational complexity of the first part is:
\begin{equation}
    O_1=\mathcal{O}(I)+\mathcal{O}(d^2I)+\mathcal{O}(K{log}_2K).
\end{equation}

The second part is  the  computation of  $\mathbf{P}^t$ and $\mathbf{P}_\mathcal{P}^t$ . First, computing $\mathbf{\Pi}(j,r)$ according to \eqref{index_matrix} needs $\mathcal{O}(JR)$ operations. Then, computation of $\mathbf{P}_\mathcal{P}^t$  and  $\mathbf{P}^t$ needs $\mathcal{O}(d^2J)$ and $\mathcal{O}(d^2J+d^2R)$ multiplications respectively. The total computational complexity of the second part is therefore:
\begin{equation}
    O_2=\mathcal{O}(JR)+\mathcal{O}(d^2J)+\mathcal{O}(d^2J+d^2R).
\end{equation}

The third part is the computation of $\mathbf{M}_{t}$. Both of the eigenvalue decomposition and the updating procedure consume $\mathcal{O}(d^3)$ multiplications. Therefore, the third part has a complexity of:
\begin{equation}
    O_3=\mathcal{O}(d^3).
\end{equation}

Summing the above three parts gives \eqref{eq:complexity}, which completes the proof.
\end{proof}

Because the value of $d$ is relatively small, the overall computational complexity depends mainly on the complexity of computing the $\mathbf{\Pi}(j,r)$ matrix, which is quadratic with respect to $\mathcal{P}$, $\mathcal{N}_0$. The complexity of computing $\mathbf{\Pi}(j,r)$ is reduced by the random sampling strategy described in Section \ref{subsec:opt_alg}, which leads the following corollary.

\begin{myCor}\label{cal}
  Given the batch size $s$, the computational complexity of pAUCMetric is reduced to
  \begin{equation}
    \mathcal{O}(2cs^3)
  \end{equation}
  where $c$ is a coefficient related to the FPR range $[\alpha,\beta]$.
\end{myCor}
\begin{proof}
  According to Section \ref{subsec:opt_alg}, we have $I=2s^2-s$, $J=s$, $K=2s^2-2s$, and $R=c(2s^2-2s)$. Therefore, the computational complexity is reduced to $\mathcal{O}(2cs^3)+\mathcal{O}(d^3)$. Because the dimension $d$ is small, the computational complexity depends mainly on $s$ only.
\end{proof}

Corollary \ref{cal} shows that the computational complexity of pAUCMetric is  cubic with respect to $s$. As will be shown in Section \ref{sec:Complexity}, pAUCMetric can achieve good performance with  a small value of $s$.

\section{The input features  of pAUCMetric}
\label{sec:feature}

After the feature extraction by a front-end, one needs to preprocess the features for boosting the performance of pAUCMetric as shown in Fig. \ref{fig:back_end_framework}. This section presents two preprocessing techniques.

\subsection{Length-normalization}

Given a speaker feature $\mathbf{y}$ from a front-end, we use the length-normalized feature \cite{garcia2011analysis} $\mathbf{x}$ as the input to pAUCMetric:
\begin{equation}
  \mathbf{x} = \frac{\mathbf{y}}{\|\mathbf{y}\|_2}
\end{equation}
The underlying reason for this normalization is as follows.
Learning a transform matrix in the cosine similarity scoring framework, i.e.,
 \begin{equation}
\label{eq:cosine_scores}
    S_{\rm cos}(\mathbf{y}_{1},\mathbf{y}_{2}; \mathbf{M})=\frac{\langle \mathbf{A}\mathbf{y}_{1},\mathbf{A}\mathbf{y}_{2} \rangle}{\|\mathbf{A}\mathbf{y}_{1} \|_2\|\mathbf{A}\mathbf{y}_{2}\|_2}
\end{equation}
has been studied extensively, e.g. \cite{dehak2011front}.
However, the learning problem is nonlinear and non-convex. Existing methods either learn $\mathbf{A}$ independently by, e.g., LDA, WCCN \cite{dehak2011front}, or learn $\mathbf{A}$ in the above framework with a good initialization \cite{bai2018cosine}. Both ways are suboptimal.

The Euclidean distance scoring is empirically inferior to the cosine similarity scoring when given the same input $\mathbf{y}$. But it is equivalent to the cosine similarity scoring if its input is the length-normalized feature $\mathbf{x}$ since
\begin{align}
\label{euclid_cosine}
   S_{\rm Euc}(\mathbf{x}_{1},\mathbf{x}_{2}) &=\|\mathbf{x}_{1}\|_2^2+\|\mathbf{x}_{2}\|_2^2-2 \langle \mathbf{x}_{1}, \mathbf{x}_{2} \rangle \nonumber \\
   &=2-2S_{\rm cos}(\mathbf{y}_{1},\mathbf{y}_{2}).
\end{align}
More importantly, learning $\mathbf{A}$ in the following Euclidean distance scoring framework does not suffer from the nonlinear and non-convex issues:
\begin{align}\label{eq:mahalanobis}
   S_{\rm Euc}(\mathbf{x}_{1},\mathbf{x}_{2};\mathbf{A}) & = \|\mathbf{A} \mathbf{x}_{1}-\mathbf{A} \mathbf{x}_{2}\|_2^2  \nonumber\\ &=(\mathbf{x}_{1}-\mathbf{x}_{2})^T\mathbf{M}(\mathbf{x}_{1}-\mathbf{x}_{2})\nonumber\\
   &=S(\mathbf{x}_{1},\mathbf{x}_{2};\mathbf{M})
\end{align}
where $\mathbf{M}=\mathbf{A}^T\mathbf{A}$ and $S(\cdot)$ is the scoring function of our pAUCMetric.

\subsection{PLDA-based preprocessing}

\begin{table*}[!t]
\renewcommand{\arraystretch}{1.2}
   \caption{Parameter settings of front-ends. The terms ``dim'' and ``mix'' is short for \textit{dimensions} and \textit{mixtures} respectively. The terms $\Delta$ and $\Delta\Delta$ denote the delta and delta-delta coefficients of MFCCs respectively.}
   \label{tab:swbd_sre}
   \centering
\scalebox{1} {
   \begin{tabular}{lcccc}
   \hline
   \hline
   Systems&& i-vector && x-vector \\
   \hline
    8 kHZ system && 20-dim MFCCs +$\Delta$ + $\Delta\Delta$ / 2048-mix GMM / 600-dim i-vector  &&  23-dim MFCCs / 512-dim x-vector\\
   16 kHZ system && 24-dim MFCCs +$\Delta$ + $\Delta\Delta$ / 2048-mix GMM / 400-dim i-vector  &&  30-dim MFCCs / 512-dim x-vector\\
  \hline
  \hline
\end{tabular}
}
\end{table*}

Two kinds of PLDA algorithms have been widely adopted in speaker verification, i.e., the simplified PLDA \cite{garcia2011analysis, ioffe2006probabilistic} and the two-covariance based PLDA \cite{brummer2010speaker}. We adopt the latent variables of the simplified PLDA \cite{ioffe2006probabilistic} as the input features of pAUCMetric.
 It generates a centralized feature $\mathbf{x}$ by first generating a speaker center $\mathbf{h}$ according to:
\begin{equation}
   \mathbf{h} \sim \mathcal{N}(0,\mathbf{\Phi}_b)
\end{equation}
and then generating the observation data according to:
\begin{equation}
   \mathbf{x} \sim \mathcal{N}(\mathbf{h},\mathbf{\Phi}_w),
\end{equation}
where $\mathbf{\Phi}_b$  is required to be positive semi-definite, and $\mathbf{\Phi}_w$  is required to be positive definite. The expectation maximization algorithm is employed to estimate the parameters. $\mathbf{\Phi}_b$ and $\mathbf{\Phi}_w$   can be simultaneously diagonalized by solving the following generalized eigenvalue  problem:
\begin{equation}
   \label{eq:LDA2}
   \mathbf{\Phi}_b \mathbf{w} =\psi \mathbf{\Phi}_w \mathbf{w},
\end{equation}
which leads to
\begin{eqnarray}
  &&\mathbf{W}\mathbf{\Phi}_b\mathbf{W}^T=\mathbf{\Psi}\\
   &&\mathbf{W} \mathbf{\Phi}_w \mathbf{W}^T=\mathbf{I}_0
\end{eqnarray}
where $\mathbf{W}$ is a square matrix whose columns are the generalized eigenvectors of \eqref{eq:LDA2}, $\mathbf{\Psi}$ is a diagonal matrix whose diagonal elements are the generalized eigenvalues of \eqref{eq:LDA2}, and $\mathbf{I}_0$ is the identity matrix.

Finally, the centralized feature is calculated as:
\begin{equation}\label{PLDA_mode}
   \mathbf{x} = \mathbf{W}^{-1}\mathbf{u},
\end{equation}
where $\mathbf{u} \sim \mathcal{N}(\mathbf{v},\mathbf{I}_0)$, and $\mathbf{v} \sim \mathcal{N}(\mathbf{v},\mathbf{\Psi})$, $\mathbf{v}$ represents the speaker, and $\mathbf{u}$ represents an example of that speaker in the latent space. Therefore, the example $\mathbf{x}$ in the original space is related to its latent representation  $\mathbf{u}$ via an invertible transformation $\mathbf{W}$.  We take the latent variable $\mathbf{u}$ as input features of pAUCMetric\footnotemark[1]. This preprocessing method adopts the advantage of the PLDA adaptation into the input features, which improves the overall performance of the speaker verification system.

\footnotetext[1]{Similar to the implementation of the PLDA  in Kaldi, we normalize $\mathbf{u}$ to $ \mathbf{u} \times \sqrt{\frac{ d}{\mathbf{u}^T(\mathbf{\Psi}+\mathbf{I}_0)^{-1}\mathbf{u}}}$, where $d$ is the dimension of $\mathbf{u}$.}

\section{Experiments} \label{sce:experiment}

\begin{table}[!t]
\renewcommand{\arraystretch}{1.2}
   \caption{Descriptions of training datasets. }
   \label{tab:traindata}
   \centering
\scalebox{0.9} {
   \begin{tabular}{lccc}
   \hline
   \hline
   & SWBD & SRE & VoxCeleb                \\
   \hline
   Languages        & English     & English (most) , others         & Multilingual                       \\
   \#Speakers       & 2,594       & 4,392                           & 7,363                         \\
   \#Recordings     & 28,181      & 64,388                          & 1,281,762                       \\
   Data sources     & Telephone   & Telephone, microphone           & Multi-media                    \\
   Environments     & Clean       & Clean                           & Real world noise     \\
  \hline
  \hline
\end{tabular}
}
\end{table}

\begin{table}[!t]
\renewcommand{\arraystretch}{1.2}
   \caption{Descriptions of evaluation datasets.}
   \label{tab:testdata}
   \centering
\scalebox{0.95} {
   \begin{tabular}{lccc}
   \hline
   \hline
   & SRE16 & SITW \\
   \hline
   Enrollment durations      & 60 $\sim$ 180 secs         & 6 $\sim$ 180 secs        \\
   Test durations            & 10 $\sim$ 60 secs          & 6 $\sim$ 180 secs        \\
   Data sources              & Telephone                  & Multi-media              \\
   Evaluation kinds          & Cantonese / Tagalog        & Dev.Core / Eval.Core     \\
  \#Evaluation trials       & 965,396 / 1,021,332         & 338,226 / 721,788              \\
  \hline
  \hline
\end{tabular}
}
\end{table}

In this section, we first present the datasets and  experimental settings and then the main results as well as analysis on the effects of the hyperparameters of pAUCMetric.

\subsection{Datasets}\label{subsec:data}

\subsubsection{Training datasets}

The training data consists of Switchboard (SWBD), NIST speaker recognition evaluation (SRE), and VoxCeleb database.
SWBD  consists of Switchboard Cellular 1 and 2 as well as Switchboard 2 Phase 1, 2, and 3. It contains $28,181$ English utterances from $2,594$ speakers. The SRE database consists of NIST SREs from 2004 to 2010 along with Mixer 6. It contains  $64,388$ telephone and microphone recordings from  $4,392$ speakers. Most of the utterances are in English, while some utterances are in Chinese, Russian, Arabic etc.
VoxCeleb consists of VoxCeleb1 \cite{Nagrani2017} and VoxCeleb2 \cite{Chung2018}, which contains over 1 million recordings from 7,363 celebrities. It is collected from real world noisy environments, therefore it contains background chatter, laughter and overlapping speech etc. In addition, we adopt the same data augmentation scheme as in \cite{snyder2018xvector} to further increase the amount and diversity of the training data. See Table \ref{tab:traindata} for summarization of the training data.

\subsubsection{Evaluation datasets}

The evaluation data include NIST SRE 2016 (SRE16) \cite{sre2016} and the Speakers in the Wild (SITW) \cite{mclaren2016speakers} datasets. Specifically, SRE16  contains two major languages---Cantonese and Tagalog. They are recorded in real-world noisy environments. The Cantonese language contains $965,393$ trials. The Tagalog language contains $1,021,332$ trials. The enrollment segments vary from 60 to 180 seconds, and the test utterances are about 10 to 180 seconds long.
The SITW is collected from open-source media, which contains real-world noise, reverberation, intra-speaker variability and compression artifacts. It contains 299 speakers. Each recording varies from 6 to 180 seconds. It has two evaluation tasks---Dev.Core which consists of $338,226$  trials, and Eval.Core, which consists of $721,788$ trials. See Table \ref{tab:testdata} for summarization of the evaluation data.

\begin{table}[!t]
\renewcommand{\arraystretch}{1.2}
   \caption{Output dimensions of the LDA in the back-ends, which are the default values of Kaldi.}
   \label{tab:LDA_par}
   \centering
\scalebox{1} {
   \begin{tabular}{lcccc}
   \hline
   \hline
   Systems & &i-vector & &x-vector \\
   \hline
    8 kHZ system  && 200 dim  & & 150 dim \\
    16 kHZ system && 150 dim  & & 128 dim \\
  \hline
  \hline
\end{tabular}
}
\end{table}

\begin{table*}[!t]
   \renewcommand{\arraystretch}{1.3}
   \caption{Comparison results of pAUCMetric and PLDA in the E1 evaluation scheme.}
   \label{tab:sre16}
   \centering
   \scalebox{0.83} {
   \begin{tabular}{llccccccccccccc}
   \hline
   \hline
   &  &  \multicolumn{5}{c}{i-vector} & & \multicolumn{5}{c}{x-vector}\\
   \cline{3-7} \cline{9-13}
   & Back-ends & EER(\%) & $\rm{minDCF}$  & $\rm pAUC_{[0,0.01]}$ & $\rm{AUC}$ &AP(\%) & & EER(\%)&$\rm{minDCF}$   &$\rm pAUC_{[0,0.01]}$ & $\rm{ACU}$&AP(\%) \\
   \hline
   \multirow{3}{*}{Cantonese}
   &PLDA            & 10.29        & 0.654        &   0.570       & 0.964       &69.38  &&  6.78       &   0.531       &0.689        & 0.982 & 80.24  \\
   &TripletMetric   & 10.22        & 0.667        &    0.559      & 0.965       &68.62  &&  6.42       &   0.527       &0.695        & 0.983 & 80.93   \\
   &pAUCMetric      & \textbf{9.52}&\textbf{0.649}&\textbf{0.578} &\textbf{0.969}&\textbf{70.58}&&\textbf{6.00}& \textbf{0.503}&\textbf{0.717}& \textbf{0.986}&\textbf{82.60} \\
  \hline
  \multirow{3}{*}{Tagalog}
   &PLDA            &\textbf{21.39}&\textbf{0.985}&\textbf{0.178} &\textbf{0.864}&\textbf{23.03}&&\textbf{18.34}& \textbf{0.977}&\textbf{0.218}& \textbf{0.894} &\textbf{28.90}  \\
   &TripletMetric  & 22.05        &0.985        & 0.175          & 0.859        &22.36  &&  18.42       & 0.980         &0.216         &0.894   &28.21   \\
   &pAUCMetric     & 21.85        & 0.985        & 0.175         &  0.860       &22.45&&  18.52       &  0.980        &0.218         & 0.894  & 28.39    \\
   \hline
   \hline
\end{tabular}
}
\end{table*}

\begin{table*}[!t]
   \renewcommand{\arraystretch}{1.3}
   \caption{Comparison results of pAUCMetric and PLDA-adp in the  E2 evaluation scheme.}
   \label{tab:adpsre16}
   \centering
   \scalebox{0.83} {
   \begin{tabular}{llccccccccccccc}
   \hline
   \hline
    &  &  \multicolumn{5}{c}{i-vector} & & \multicolumn{5}{c}{x-vector}\\
   \cline{3-7} \cline{9-13}
   & Back-ends & EER(\%) &$\rm{minDCF}$ &$\rm pAUC_{[0,0.01]}$ & $\rm{AUC}$  & AP(\%) & & EER(\%)  & $\rm{minDCF}$  &$\rm pAUC_{[0,0.01]}$ & $\rm{AUC}$ &AP(\%)\\
   \hline
   \multirow{3}{*}{Cantonese}
   &PLDA-adp           & 8.91         & 0.597        & 0.625        & 0.970  &74.35   && 4.80         &0.400          & 0.800   & 0.990   &88.26    \\
   &TripletMetric      & 8.21         & 0.586        & 0.641        & 0.976  &76.00   && 4.34         & 0.391         & 0.810   &0.992      &88.98        \\
   &pAUCMetric         & \textbf{7.93}&\textbf{0.577}&\textbf{0.646}&\textbf{0.977}&\textbf{76.60} && \textbf{4.19}&\textbf{0.379} & \textbf{0.818} & \textbf{0.993} &\textbf{89.58} \\
  \hline
   \multirow{3}{*}{Tagalog}
   &PLDA-adp           & 19.85        &0.892& 0.313        & 0.885   & 39.00   && 12.27        & 0.753 & 0.499      & 0.948   & 60.21   \\
   &TripletMetric      &\textbf{18.95} &\textbf{0.892}&\textbf{0.322}&\textbf{0.894}&\textbf{40.23}&&12.04& \textbf{0.750}& \textbf{0.506}& 0.950& \textbf{60.80}   \\
   &pAUCMetric         &19.11 & 0.896        &0.315&0.892&39.48 &&\textbf{11.97}& 0.754       & 0.503 &\textbf{0.951}&60.59 \\
   \hline
   \hline
\end{tabular}
}
\end{table*}

\subsection{Experimental settings} \label{subsec:baseline}

\subsubsection{Training schemes}
Due to different collection methods and sampling rates of the training data, we define two kinds of systems:

\begin{itemize}
  \item \textbf{8 kHZ system:} We adopt the augmented SWBD and SRE data, which include $220,569$ recordings in total, to train front-end feature extractors.  The back-ends are trained on the augmented SRE data. The signals originally sampled at 16 kHz are downsampled to 8 kHz.
  \item \textbf{16 kHZ system:} We use the VoxCeleb data to train an i-vector feature extractor, and use the augmented VoxCeleb data to train an x-vector feature extractor. We randomly selected $200,000$ recordings from the augmented VoxCeleb data to train the back-ends.
\end{itemize}

\subsubsection{Front-ends}
We use the GMM-UBM/i-vector and x-vector front-ends to extract speaker features. The front-ends are implemented using Kaldi \cite{povey2011kaldi}. Their parameter settings are also the same as in Kaldi, which are summarized in Table \ref{tab:swbd_sre}.

Specifically, for the i-vector extractor, the frame length is 25 ms, and the frame shift is 10 ms. The frame-level acoustic features of the 8 kHZ and 16 kHZ systems are 20- and 24-dimensional MFCCs respectively, which are further mean-normalized over a sliding window of 3 s. The final acoustic features are a concatenation of the MFCCs and their delta and delta-delta coefficients, which produces a total of 60-dimensional acoustic feature vector for the 8 kHZ system and 72-dimensional acoustic feature vector for the 16 kHZ system. An energy-based voice activity detector (VAD) is employed to remove non-speech frames. The number of Gaussian mixtures is set to 2048 for both the 8 kHZ and 16 kHZ systems. The dimension of the i-vectors is set to 600 for the 8 kHZ system, and 400 for the 16 kHZ system.

For the x-vector extractor, we used the standard Kaldi SRE16 and SITW recipes. Specifically, the frame-length is 25 ms, and the frame shift is 10 ms. The acoustic features of the 8 kHZ and  16 kHZ systems are 24- and 30-dimensional MFCCs, respectively, which are further mean-normalized over a sliding window of 3 s. The energy-based VAD is the same as that in the i-vector extractor. The 8 kHZ x-vector extractor is a pre-trained system provided at \textit{http://kaldi-asr.org/models/m3}. The 16 kHZ x-vector extractor is a newly trained system by Kaldi. The dimensions of the x-vectors in both the systems are set to 512.

\subsubsection{Back-ends}
We compare pAUCMetric with the state-of-art PLDA back-end and a commonly used cosine similarity scoring back-end. The parameter settings of the compared back-ends are summarized as following.

\begin{table*}[!t]
   \renewcommand{\arraystretch}{1.3}
   \caption{ Comparison results of pAUCMetric and PLDA in the  E3 evaluation scheme.}
   \label{tab:sitw_sre}
   \centering
   \scalebox{0.83} {
   \begin{tabular}{llccccccccccccc}
   \hline
   \hline
   &&  \multicolumn{5}{c}{i-vector} & & \multicolumn{5}{c}{x-vector} \\
   \cline{3-7} \cline{9-13}
   &Back-ends&EER(\%)&$\rm{minDCF}$& $\rm pAUC_{[0,0.01]}$ &$\rm{AUC}$ &AP(\%)&& EER(\%) & $\rm{minDCF}$ &$\rm pAUC_{[0,0.01]}$& $\rm{AUC}$&AP(\%)\\
   \hline
   \multirow{3}{*}{Dev.Core}
   & PLDA          & 9.20         & 0.619         & 0.590        & 0.972     &62.01&&6.85         & 0.513          &  0.697        & 0.985  & 72.31       \\
   &TripletMetric  & 9.70         & 0.624         & 0.583        & 0.970     &61.05     && 6.43        & 0.515          &  0.697        & 0.986  & 72.39 \\
   &pAUCMetric     &\textbf{8.73} &\textbf{0.605} &\textbf{0.600}&\textbf{0.974}&\textbf{62.99}&&\textbf{5.91}& \textbf{0.500} & \textbf{0.724}  & \textbf{0.988} &\textbf{74.87}  \\
   \hline
   \multirow{3}{*}{Eval.Core}
   &PLDA           &10.03         & 0.646        &0.563          & 0.969     &54.70&& 6.75        & 0.546          &0.674            & 0.984    & 66.13     \\
   &TripletMetric  &10.14         & 0.654        &0.552          & 0.968     &53.77 && 6.86        & 0.562         & 0.669     &0.984         &65.38     \\
   &pAUCMetric    &\textbf{9.65} &\textbf{0.639}&\textbf{0.571} &\textbf{0.970}&\textbf{55.73}&&\textbf{6.10}&\textbf{0.528}  &\textbf{0.703}   & \textbf{0.986} &\textbf{68.71} \\
  \hline
  \hline
\end{tabular}
}
\end{table*}

\begin{table*}[!t]
   \renewcommand{\arraystretch}{1.3}
   \caption{Comparison results of pAUCMetric and PLDA in the E4 evaluation scheme.}
   \label{tab:sitw_vox}
   \centering
   \scalebox{0.83} {
   \begin{tabular}{llccccccccccccc}
   \hline
   \hline
   &&  \multicolumn{5}{c}{i-vector} & & \multicolumn{5}{c}{x-vector}\\
   \cline{3-7} \cline{9-13}
   & Back-ends & EER(\%)&$\rm{minDCF}$&$\rm pAUC_{[0,0.01]}$& $\rm{AUC}$   &AP(\%) && EER(\%) &$\rm{minDCF}$&$\rm pAUC_{[0,0.01]}$ &$\rm{AUC}$& AP(\%)\\
   \hline
   \multirow{3}{*}{Dev.Core}
   &PLDA          & 5.30           & \textbf{0.418}& 0.772      & 0.989     &79.90 && 2.96         & 0.301        & 0.868        &  0.996       &88.69 \\
   &TripletMetric & 5.27           & 0.429        & 0.765     & 0.989      &79.32  && 2.77         & 0.309        &0.862         &  0.996  & 88.22     \\
   &pAUCMetric    &\textbf{4.93}   & 0.420   &\textbf{0.776}   &\textbf{0.990}& \textbf{80.3}&&  \textbf{2.58}&\textbf{0.289}&\textbf{0.880}&\textbf{0.997}& \textbf{89.71}  \\
   \hline
   \multirow{3}{*}{Eval.Core}
   &PLDA           & 5.72           & \textbf{0.453}& 0.746           &0.988    &73.91 &&  3.58        &0.333        & 0.847        & 0.995        &84.23\\
   &TripletMetric  & 6.04           & 0.464         & 0.738           &0.987    &73.10 &&  3.68        &0.341        &0.842      &0.995& 83.59    \\
   &pAUCMetric     &\textbf{5.49}   & 0.456         &\textbf{0.748}   &\textbf{0.988}& \textbf{74.32} && \textbf{3.23}&\textbf{0.316}&\textbf{0.861}&\textbf{0.996}&\textbf{85.43}     \\

  \hline
  \hline
\end{tabular}
}
\end{table*}

\begin{itemize}

\item \textbf{PLDA:} We first reduce the speaker features into a low dimensional vector by  linear discriminant analysis (LDA). Specifically, if the i-vector front-end is used, the LDA dimension is set to 200 for the 8 kHZ system and  150 for the 16 kHZ system. If the x-vector front-end is used, the LDA dimension is set to 150 and 128 in the  8 kHZ  system and 16 kHZ system, respectively. The dimensions of LDA are summarized in Table \ref{tab:LDA_par}. We use the output of LDA as the input of PLDA to compute the similarity scores.

    \item \textbf{Cosine similarity scoring (Cosine):} We adopt the same dimension reduction as that in Table \ref{tab:LDA_par} by LDA, and then use the dimension-reduced feature as the input of the cosine similarity scoring to make decisions.

\item \textbf{PLDA-adp:}
We conduct domain adaptation to the PLDA back-end of the 8 kHZ system by using an unlabeled major dataset in
NIST 2016 SRE, which consists of $2,272$ utterances. The adaptation technique is implemented in “kaldi-master/egs/sre16”
of Kaldi.

\item \textbf{pAUCMetric:}
We adopt the same dimension reduction as that in Table \ref{tab:LDA_par} by LDA. Then, the speaker features are preprocessed according to Section \ref{sec:feature}. At last, the preprocessed features are used as the input of pAUCMetric.
The default hyperparameters of pAUCMetric are as follows. $\alpha=0$, $\beta=0.01$, $\mu=10^{-3}$, $\eta=10$, and $s=500$. $\gamma$ is set to $0.5$ for the x-vector front-end, and set to $\gamma=0.1$ for the i-vector front-end. As will be shown in Section \ref{sec:parameter}, pAUCMetric performs robustly with a wide range of hyperparameter settings.

\item {\textbf{TripletMetric:} The algorithm was proposed in the last paragraph of Section \ref{sec:trip}.  Its hyperparameter setting is the same as that of pAUCMetric.}

\end{itemize}

We evaluated the studied methods using the calibration-insensitive metrics, including the EER, $\rm{minDCF}$ with $P_{tar}=0.01$ and equal costs of misses and false alarms, {pAUC} with $\alpha=0$, $\beta=0.01$ ($\rm pAUC_{[0,0.01]}$), AUC, and average precision (AP).

We also conducted an experiment in Section \ref{subsec:cal}, where we evaluated the performance by the calibration-sensitive metrics, including the actDCF with $P_{tar}=0.01$  and equal costs of misses and false alarms, and $\rm C_{llr}$.

\subsection{Results based on PLDA-based preprocessing } \label{subsec:plda}
This section presents the main experimental results of the pAUCMetric with the PLDA-based preprocessing technique.  We evaluate both the 8 kHZ and 16 kHZ systems on the SRE16 and SITW datasets, which contains the following four evaluation schemes:
\begin{itemize}
  \item \textbf{E1:} This scheme conducts the comparison on language mismatched conditions. The evaluation is carried out with the 8 kHZ system on the SRE16 dataset. Most training data of the 8 kHZ system are in English, while the SRE16 test data are in Cantonese and Tagalog languages.

  \item \textbf{E2:} Contrary to E1, this scheme conducts the comparison on language matched conditions. The evaluation is carried out with the 8 kHZ system on the SRE16 dataset as well, and furthermore, the domain adaptation technique is adopted. The input features of pAUCMetric are the latent variables of PLDA-adp.

  \item \textbf{E3:} This scheme makes an evaluation on channel and noise mismatched conditions. We conducted the evaluation with the 8 kHZ system on the SITW data, where we downsampled the SITW from 16 KHZ to 8 KHZ. The mismatched problem is caused by the fact that SITW is collected from multi-media videos, while the training data, i.e., SWBD and SRE, are collected from telephone or meeting conditions.

  \item \textbf{E4:} Contrary to E3, this scheme makes an evaluation on channel and noise matched conditions. Specifically, we make the evaluation with the 16 kHZ system on the SITW dataset. Both the SITW and VoxCeleb datasets are collected from multi-media videos.

\end{itemize}

The experimental results of E1 are presented in Table \ref{tab:sre16}. As seen, pAUCMetric achieves obvious performance improvement over PLDA  on the Cantonese language. Specifically, when the x-vector front-end is used, it obtains $11\%$ relative EER reduction and $5\%$ relative $\rm{minDCF}$  reduction; it also achieves $9\%$ relative $\rm pAUC_{[0,0.01]}$ improvement, $22\%$ relative $\rm{AUC}$ improvement, and $12\%$ relative $\rm AP$ improvement. When the i-vector front-end is used, it obtains $7\%$ relative EER reduction and $14\%$ relative $\rm AUC$ improvement. However, the experimental results of PLDA and pAUCMetric on the Tagalog language are not good, which may be due to the large mismatch between the Tagalog and the languages of the training data, as well as the fact that the Tagalog data is quite noisy.

\begin{figure}[t!]
  \centering
  \includegraphics[width=2.8in]{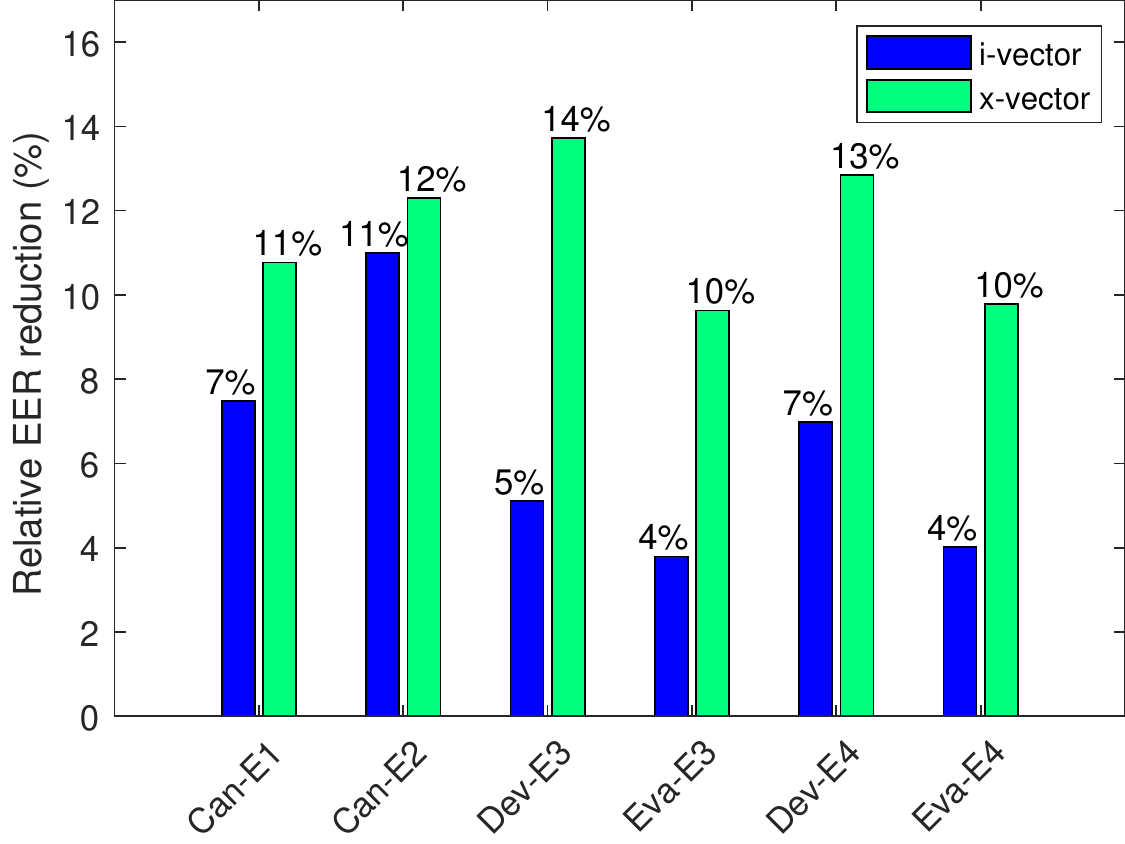}
  \caption{Relative EER reduction of pAUCMetric over PLDA. The terms ``Can'', ``Dev'', and ``Eva'' denote the Cantonese data of SRE16, the Dev.Core and Eval.Core tasks of SITW, respectively.  The terms ``E1'', ``E2'', ``E3'', and ``E4'' denote the four evaluation schemes.
}
  \label{fig:relative_eer_improvement}
\end{figure}

The experimental results of E2 are presented in Table \ref{tab:adpsre16}.  It is seen that pAUCMetric yields better performance than PLDA-adp, for both the i-vector and x-vector front-ends. Specifically, when the x-vector front-end is applied to the Cantonese language of SRE16, pAUCMetric obtains $13\%$ relative EER reduction and $5\%$ relative $\rm{minDCF}$ reduction respectively; it also achieves $9\%$ relative improvement in terms of $\rm pAUC_{[0,0.01]}$ and $30\%$ relative improvement in terms of $\rm{AUC}$. When the i-vector front-end is applied to the Cantonese language, pAUCMetric also obtains $11\%$ relative EER reduction and $23\%$ relative $\rm AUC$ improvement, respectively. pAUCMetric also achieves better performance than PLDA-adp on the Tagalog language.

\begin{figure}[t!]
  \centering
  \includegraphics[width=3.5in]{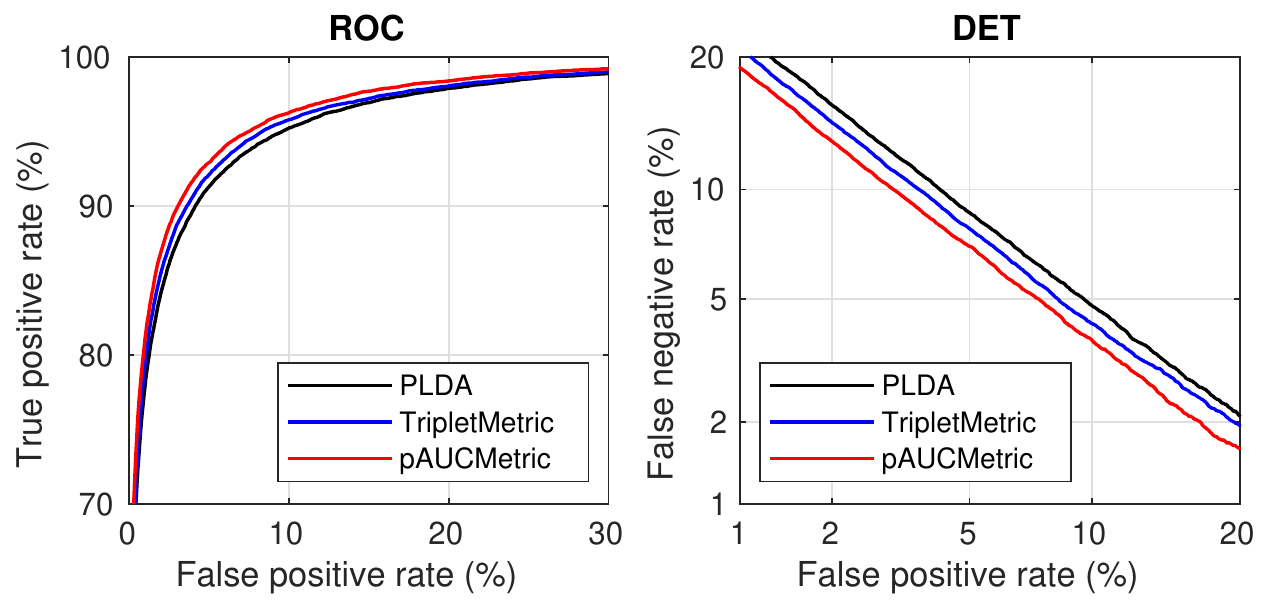}
  \caption{ROC and DET curves of the comparison methods with the x-vector front-end on the Cantonese data of SRE16 in the E1 evaluation scheme.}
  \label{fig:Cantonese_PLDA_det_curve}
\end{figure}

\begin{table*}[!t]
   \renewcommand{\arraystretch}{1.3}
   \caption{Comparison results of pAUCMetric and Cosine with the models of the 8 kHZ system.}
   \label{tab:SWBD-SRE-cosine}
   \centering
   \scalebox{0.83} {
   \begin{tabular}{llccccccccccccc}
   \hline
   \hline
   &&  \multicolumn{4}{c}{i-vector} & & \multicolumn{4}{c}{x-vector}\\
   \cline{3-7} \cline{9-13}
   &Back-ends&$\rm EER(\%)$ &$\rm{minDCF}$&$\rm pAUC_{[0,0.01]}$&$\rm{AUC}$  &AP(\%)&&EER(\%)&$\rm{minDCF}$&$\rm pAUC_{[0,0.01]}$ &$\rm{AUC}$ & AP(\%)\\
   \hline
   \multirow{3}{*}{Cantonese}
   &Cosine     &13.68         &0.744         &0.467         & 0.940        &59.10&& 9.25        &  0.606       &  0.613      & 0.968         & 73.23    \\
   &TripletMetric  &11.78     &0.715         &0.507         & 0.954        &63.60 && 6.89        &  0.551       &  0.677      &0.981           & 79.34  \\
   &pAUCMetric &\textbf{10.57}&\textbf{0.689}&\textbf{0.537}&\textbf{0.962}  &\textbf{66.75}&&\textbf{6.35}&\textbf{0.523}&\textbf{0.700}&\textbf{0.984} &\textbf{81.20} \\
   \hline
   \multirow{3}{*}{Dev.Core}
   &Cosine     &11.09         &  0.650       &0.552         & 0.960       &57.77&&   9.01      & 0.573        &   0.643      & 0.974           &66.53  \\
   &TripletMetric& 10.70      & 0.638        &0.568         & 0.966       &59.37 &&   6.89      & 0.526        &  0.684       & 0.984           &71.02  \\
   &pAUCMetric&\textbf{9.07}&\textbf{0.616}&\textbf{0.593}&\textbf{0.971}&\textbf{62.29}&&\textbf{5.94}&\textbf{0.497}&\textbf{0.719}&\textbf{0.987}&\textbf{74.41} \\
   \hline
   \multirow{3}{*}{Eval.Core}
   &Cosine     &11.86         & 0.684        &  0.519       & 0.956        &49.13&&  8.53       &  0.600       &  0.617       & 0.974       &59.92 \\
   &TripletMetric&10.83       &0.692         & 0.527        & 0.962        & 50.07    &&  7.20       &  0.576       &  0.648       & 0.982       &62.82  \\
   &pAUCMetric &\textbf{10.00}&\textbf{0.668}&\textbf{0.553}&\textbf{0.966}&\textbf{52.88}&&\textbf{6.40}&\textbf{0.539}&\textbf{0.689}&\textbf{0.985}&\textbf{67.09} \\

  \hline
  \hline
\end{tabular}
}
\end{table*}
\begin{table*}[!t]
   \renewcommand{\arraystretch}{1.3}
   \caption{Comparison results of pAUCMetric and  Cosine with the models of the 16 kHZ system.}
   \label{tab:sitw_vox-cosine}
   \centering
   \scalebox{0.83} {
   \begin{tabular}{llccccccccccccc}
   \hline
   \hline
   &&  \multicolumn{5}{c}{i-vector} & & \multicolumn{5}{c}{x-vector}\\
   \cline{3-7} \cline{9-13}
   & Back-ends &$\rm EER(\%)$&$\rm{minDCF}$ &$\rm pAUC_{[0,0.01]}$ & $\rm{AUC}$ &AP(\%)&& EER(\%) & $\rm{minDCF}$ & $\rm pAUC_{[0,0.01]}$ & $\rm{AUC}$&AP(\%)\\
   \hline
   \multirow{3}{*}{Dev.Core}
   &Cosine     & 7.30        &0.569         & 0.661       &  0.981      &68.20&& 4.62        & 0.472        & 0.758        & 0.991 & 77.68      \\
   &TripletMetric &6.12      & 0.532        & 0.691       &  0.987      &71.74&& 3.85        & 0.442        & 0.785        & 0.994 &  80.34    \\
   &pAUCMetric &\textbf{5.61}&\textbf{0.496}&\textbf{0.722}&\textbf{0.988}&\textbf{74.89}&&\textbf{3.35}&\textbf{0.352}&\textbf{0.834}&\textbf{0.995}&\textbf{85.35} \\
   \hline
   \multirow{3}{*}{Eval.Core}
   &Cosine     & 7.45        & 0.606        & 0.616        & 0.979     &60.60     &&  5.41       & 0.465        &  0.744       & 0.989  & 73.15  \\
   &TripletMetric& 6.40      & 0.581        & 0.644        & 0.984     & 62.94    &&  4.54       & 0.444        & 0.769        &0.992   & 75.69      \\
   &pAUCMetric &\textbf{5.96}&\textbf{0.547}&\textbf{0.679}&\textbf{0.986}&\textbf{66.51}&&\textbf{3.80}&\textbf{0.374}&\textbf{0.825}&\textbf{0.994}&\textbf{81.48}\\

  \hline
  \hline
\end{tabular}
}
\end{table*}

The experimental results of E3 are presented in Table \ref{tab:sitw_sre}. One can see that pAUCMetric achieves better performance than PLDA. Specifically, when the x-vector front-end is used, pAUCMetric achieves $10\%$ relative EER reduction, $3\%$ relative $\rm{minDCF}$ reduction, and  more than $8\%$ relative $\rm pAUC_{[0,0.01]}$ improvement on both the Dev.Core and Eval.Core tasks; it also obtains more than $20\%$ and $12\%$ relative AUC improvement on the Dev.Core and Eval.Core tasks, respectively. When the i-vector front-end is used, it also achieves better performance than PLDA.

The experimental results of E4 are presented in Table \ref{tab:sitw_vox}. From this table, one can see that pAUCMetric also yields better performance than PLDA. Specifically, when the x-vector front-end is used, it obtains approximately $10\%$ relative EER reduction; it also obtains about $9\%$ relative $\rm pAUC_{[0,0.01]}$ improvement, and more than $20\%$ relative $\rm{AUC}$ improvement. When the i-vector front-end is used, a similar experimental phenomenon is observed as well.

To summarize, when the x-vector front-end is used, pAUCMetric obtains about $10\%$ relative EER reduction, $9\%$ relative $\rm pAUC_{[0,0.01]}$ improvement, and more than $20\%$ relative $\rm{AUC}$ improvement over the state-of-the-art PLDA, except the Eval.Core task of the SITW dataset in the E3 evaluation scheme.
Although the performance improvement with the i-vector front-end is not so significant as that with the x-vector front-end, the trends are consistent. For clarity, the relative EER improvement of pAUCMetric over PLDA in different evaluation schemes is summarized in Fig. \ref{fig:relative_eer_improvement}.
pAUCMetric also achieves better performance than TripletMetric in all of the above four conditions.

Figure \ref{fig:Cantonese_PLDA_det_curve}  plots the ROC and  DET curves of the comparison methods with the x-vector front-end in the SRE16 Cantonese of the E1 evaluation scheme. It is seen from the figure that pAUCMetric yields  better  ROC and  DET curves than PLDA.
We further draw the DET curves of the $\rm E2 \sim E4$ schemes in Appendix C, where we also see the effectiveness of pAUCMetric.

\subsection{Results based on length-normalization  preprocessing}\label{subsec:cos}

\begin{table*}[!t]
   \renewcommand{\arraystretch}{1.3}
   \caption{EER results of the comparison back-ends with different input feature dimensions. The term ``length-normalization'' denotes the length normalization preprocessing. The term ``PLDA-based'' denotes the PLDA-based preprocessing.}
   \label{tab:dimension_analysis}
   \centering
   \scalebox{0.83} {
   \begin{tabular}{llcccccccc}
   \hline
   \hline
   & Back-ends &  50 dim  & 100 dim & 150 dim  & 200 dim  & 250 dim & 300 dim  & 350 dim  & 400 dim   \\
   \hline
   \multirow{2}{*}{Length-normalization }
   &Cosine      & 9.14  &  8.51  &  9.25   &  9.93  & 10.91  & 11.78  &  12.55 & 13.23    \\
   &pAUCMetric  & 8.80  &  6.85  &  6.35   &  6.50  &  6.62  & 6.90   &  7.17  & 7.43     \\
   \hline
   \multirow{2}{*}{PLDA-based}
   &PLDA        & 8.36   & 6.72   & 6.82   &  7.50   &  8.13  & 8.77   &9.26   &   9.67      \\
   &pAUCMetric  & 8.02   & 6.19   & 6.05   &  6.54   &  7.08  & 7.67   &7.89   &   8.18      \\
  \hline
  \hline
\end{tabular}
}
\end{table*}

\begin{figure}[t!]
  \centering
  \includegraphics[width=2.8in]{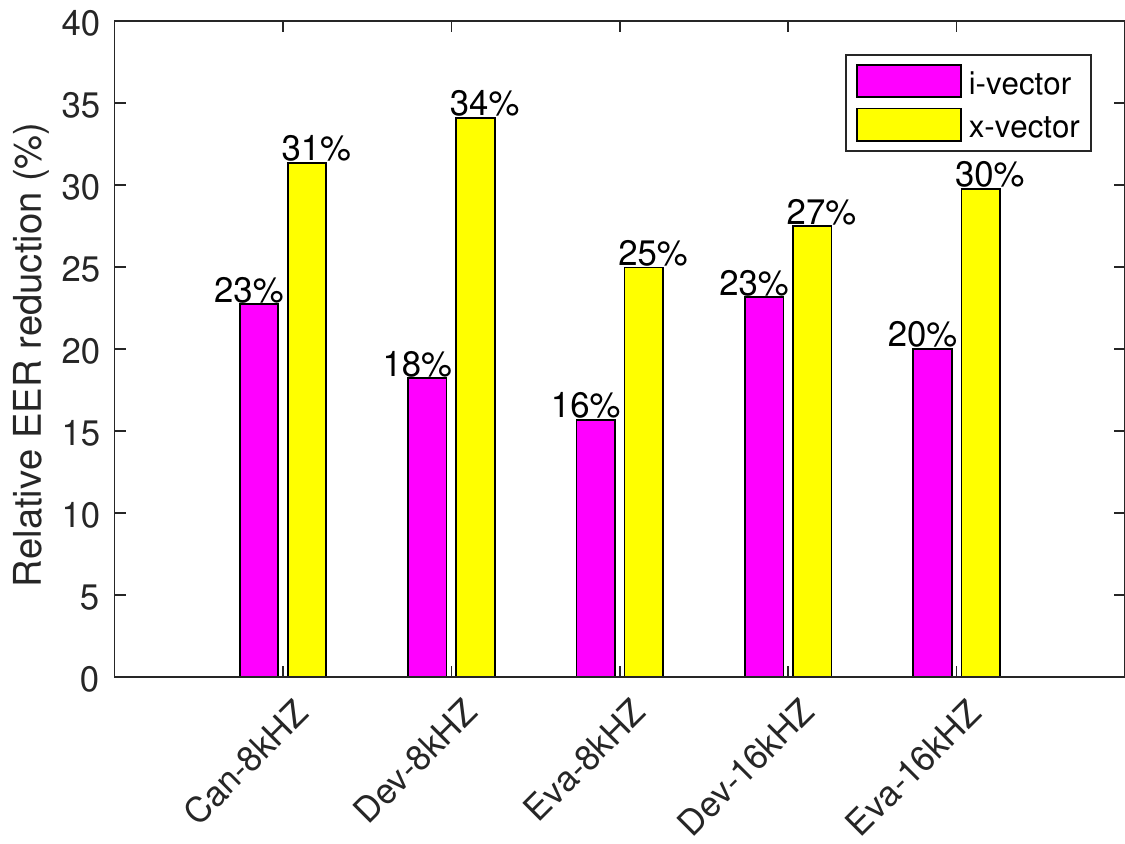}
  \caption{Relative EER reduction of pAUCMetric over Cosine. The terms ``Can'', ``Dev'', and ``Eva'' denote the Cantonese data of SRE16, the Dev.Core and Eval.Core. tasks of SITW, respectively.  The terms ``8kHZ'' and ``16kHZ'' denote the 8 kHZ system and 16 kHZ system respectively.}
  \label{fig:relative_EER_cosine}
\end{figure}

\begin{figure}[t!]
  \centering
  \includegraphics[width=3.5in]{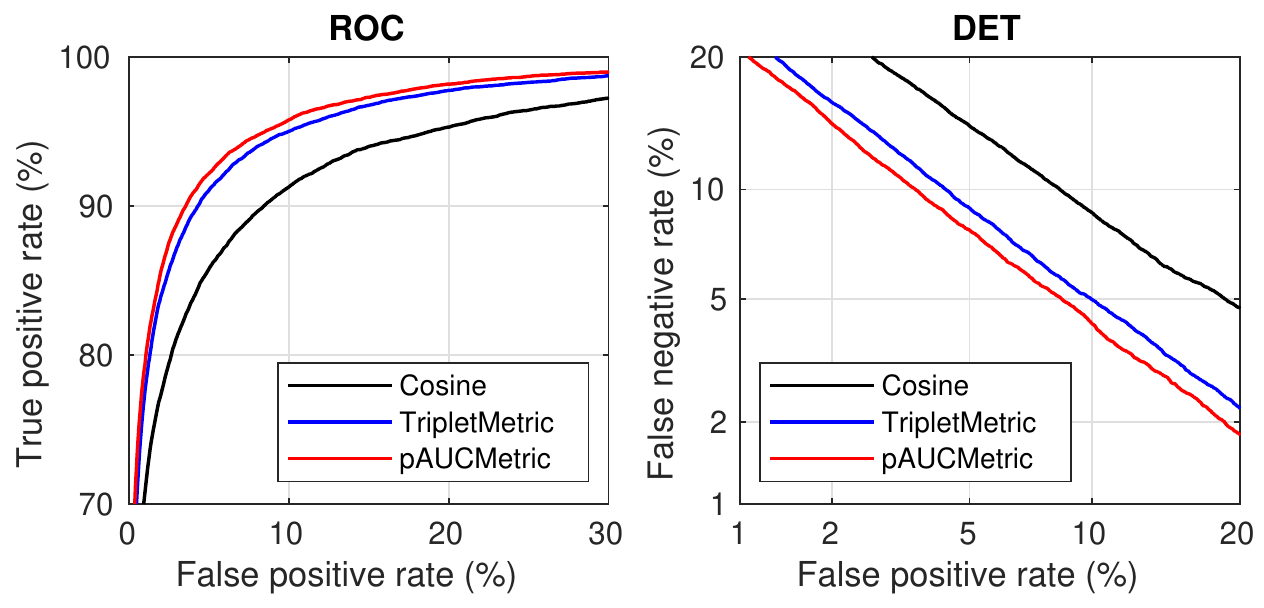}
  \caption{ROC and DET curves of the comparison methods with the x-vector front-end of the 8 kHZ system on the Cantonese data of SRE16.}
  \label{fig:Cantonese_PLDA_adp_det_curve}
\end{figure}

This section presents the main experimental results of the pAUCMetric with the length-normalization preprocessing technique. We compare it with the Cosine back-end.

Specifically, we first evaluate the 8 kHZ system on the Cantonese data of SRE16 and the Dev.Core and Eval.Core tasks of SITW. The experimental results are summarized in Table \ref{tab:SWBD-SRE-cosine}. As shown in the table, pAUCMetric achieves significant performance improvement over the Cosine back-end.
When the i-vector front-end is used, it obtains about $16\% $ to $ 23\%$ relative EER reduction, and approximately $2\% $ to $ 7\% $ $\rm{minDCF}$ reduction respectively; it also obtains about $7\% $ to $ 13\%$ relative improvement in terms of $\rm pAUC_{[0,0.01]}$, and about $23\%$ to $ 37\%$ relative improvement in terms of $\rm AUC$. When the x-vector front-end is used, pAUCMetric obtains more than  $25\%$ relative EER reduction; moreover, it obtains about $20\%$ relative $\rm pAUC_{[0,0.01]}$ improvement and more than $40\%$ relative $\rm AUC$ improvement.

Then, we evaluate the 16 kHZ system on the Dev.Core and Eval.Core tasks of SITW. The experimental results are summarized in Table \ref{tab:sitw_vox-cosine}.  One can see that pAUCMetric also yields significant performance improvement over the Cosine back-end.
For example, when the x-vector front-end is used, it obtains $27\%$ and $30\%$ relative EER reduction on the Dev.Core task and Eval.Core task respectively. It also obtains more than $40\%$ relative $\rm AUC$ improvement on both of the tasks. The performance trend with the i-vector front-end is consistent with the trend with the x-vector front-end.

To summarize, when the length-normalization is adopted to preprocess the speaker features, pAUCMetric achieves significant performance improvement over the Cosine back-end. For clarity, the relative EER improvement on different evaluation dataset is summarized in Fig. \ref{fig:relative_EER_cosine}. Moreover, the relative improvement of the pAUCMetric over PLDA with the x-vector front-end behaves better than that with the i-vector front-end. pAUCMetric also achieves better performance than TripletMetric, when the length-normalization preprocessing is adopted.

Figure \ref{fig:Cantonese_PLDA_adp_det_curve}  plots the ROC and DET curves of the comparison methods with the x-vector front-end on the SRE16 Cantonese data. It is seen from the figure that pAUCMetric yields better ROC and DET curves than Cosine.

\subsection{Calibration}\label{subsec:cal}
 In real applications, it is needed to present the verification result in terms of calibrated LLR \cite{khosravani2016aut}. So, we applied calibration to the all  the studied back-ends, i.e., Cosine, PLDA, TripletMetric, and pAUCMetric, with the linear logistic regression method of BOSARIS Toolkit\footnotemark[2], where the calibration model was trained on the Dev.Core dataset of SITW and evaluated on its Eval.Core dataset.
\footnotetext[2]{ Available at: \textit{https://sites.google.com/site/bosaristoolkit/}. }

 We conducted  experiments on the conditions of the 8 kHZ and 16 kHZ systems respectively. From the experimental results in Table \ref{tab:calibration}, one can see that pAUCMetric achieves better performance than the compared methods on all the four conditions in terms of $\rm actDCF$ and $\rm C_{llr}$.

\begin{table}[!t]
\centering
\caption{Score calibration results on the x-vector of Eval.Core task of the SITW dataset.}
\label{tab:calibration}
\scalebox{0.83} {
\begin{tabular}{llccccccc}
   \hline
   \hline

   &&\multicolumn{2}{c}{8 KHZ System} & & \multicolumn{2}{c}{16 KHZ System}\\
   \cline{3-4} \cline{6-7}
     Preprocessing&Back-ends &$\rm{actDCF}$  & $\rm C_{llr}$   && $\rm{actDCF}$ &$\rm C_{llr}$   \\
   \hline
   \multirow{3}{*}{Length-normalization }
   &Cosine        & 0.6029             &   0.3012            &&   0.4708            & 0.1961\\
   &TripletMetric  &0.5775              &   0.2590            &&   0.4494            & 0.1670     \\
   &pAUCMetric    &  \textbf{0.5404}   &   \textbf{0.2351}    &&   \textbf{0.3762}   &\textbf{0.1461} \\
   \hline
   \multirow{3}{*}{PLDA-based }
   &PLDA          & 0.5502      &0.2364             &&   0.3354      &0.1300    \\
    &TripletMetric & 0.5640      &0.2404             &&   0.3446      &0.1353      \\
   &pAUCMetric    & \textbf{0.5324}&\textbf{0.2209} && \textbf{0.3221}&\textbf{0.1199} \\

  \hline
  \hline
\end{tabular}
}
\end{table}

\subsection{Discussion} \label{subsec:discussion}
In this section, we first discuss the effect of the input feature dimension of pAUCMetric on performance, then analyze the effects of its hyperparameters, and at last discuss the computational complexity and performance with respect to the batch size $s$.

All discussions use the x-vector front-end of the 8 kHZ system to extract speaker features, and compare PLDA with the pAUCMetric that adopts the PLDA-based preprocessing on the Cantonese data of SRE16. No domain adaptation is adopted in the discussions.

\begin{table}[!t]
   \renewcommand{\arraystretch}{1.3}
   \caption{Relative EER reduction of pAUCMetric over PLDA with respect to $\gamma$ and $\mu$.}
   \label{tab:gamma_mu_EER}
   \centering
   \scalebox{0.85} {
   \begin{tabular}{l|cccccc|cc}
   \hline
   \hline
   \diagbox{$\gamma$}{$\mu$} & 0 & $10^{-7}$  &  $10^{-6}$ & $10^{-5}$ & $10^{-4}$ & $10^{-3}$ & $10^{-2}$ &$10^{-1}$  \\ 
   \hline
    0    & 7.5 & 7.4  & 7.8& 6.8 & 6.5 & 7.6 & 1.2 & -5.9 \\
    0.01 & 7.8 & 8.8  &7.8 & 7.5 & 7.6 & 7.6 & 1.7 &-6.2 \\
    0.05 & 8.4 & 9.2  &8.4 & 9.5 & 9.0 &10.3 & 2.0 &-6.1 \\
    0.10 & 9.4 & 7.4  &9.4 & 9.1 & 9.3 &10.2 & 1.7 &-6.0 \\
    0.50 & 9.6 & 9.5  &9.6 &10.8 &10.5 &11.0 & 1.9 &-6.2 \\
    1.00 & 8.0 & 10.6 &8.0 & 9.5 & 8.5 &10.9 & 2.1 &-6.0 \\
    1.50 & 6.9 & 5.5  &6.8 & 8.7 & 8.4 &8.6  & 2.7 &-5.8 \\
  \hline
   3.00  &-1.4 & -3.3 &-1.3&-3.6 &-0.2 &1.0  & 1.1 &-5.6 \\
  \hline
  \hline
\end{tabular}
}
\end{table}

\subsubsection{Effect of the input feature dimension on performance}
\begin{table}[!t]
   \renewcommand{\arraystretch}{1.3}
   \caption{Relative $\rm pAUC_{[0,0.01]}$ reduction of pAUCMetric over PLDA with respect to $\gamma$ and $\mu$.}
   \label{tab:gamma_mu_pAUC}
   \centering
   \scalebox{0.85} {
   \begin{tabular}{l|cccccc|cc}
   \hline
   \hline
   \diagbox{$\gamma$}{$\mu$}& 0 & $10^{-7}$  &  $10^{-6}$ & $10^{-5}$ & $10^{-4}$ & $10^{-3}$ & $10^{-2}$ &$10^{-1}$ \\
   \hline
       0 & 5.0 & 4.6  & 4.9  & 4.2 & 4.1 & 4.3 & -0.6 & -6.2\\
    0.01 & 5.2 & 5.4  &5.2 & 5.4 & 4.9 & 4.8 & -0.6 & -6.3 \\
    0.05 & 5.6 & 6.1  &5.6 & 6.2 & 5.8 & 6.6 & -0.1 & -6.1 \\
    0.10 & 6.9 & 5.4  &6.9 & 6.6 & 6.6 & 7.1 & -0.1 & -6.1 \\
    0.50 & 7.3 & 7.4  &7.3 & 7.3 & 7.4 & 8.0 & -0.2 & -6.1 \\
    1.00 & 6.1 & 7.3  &6.1 & 6.6 & 4.9 & 7.5 & 0.4  & -6.1 \\
    1.50 & 4.6 & 3.8  &4.6 & 2.9 & 6.0 & 5.6 & 1.5  & -5.9 \\
  \hline
   3.00  &-3.5 & -4.4 &-1.8&-2.6 &-3.2 &-0.3 &-0.7  &-6.1 \\
  \hline
  \hline
\end{tabular}
}
\end{table}

 \begin{figure}[t!]
  \centering
  \includegraphics[width=2.5in]{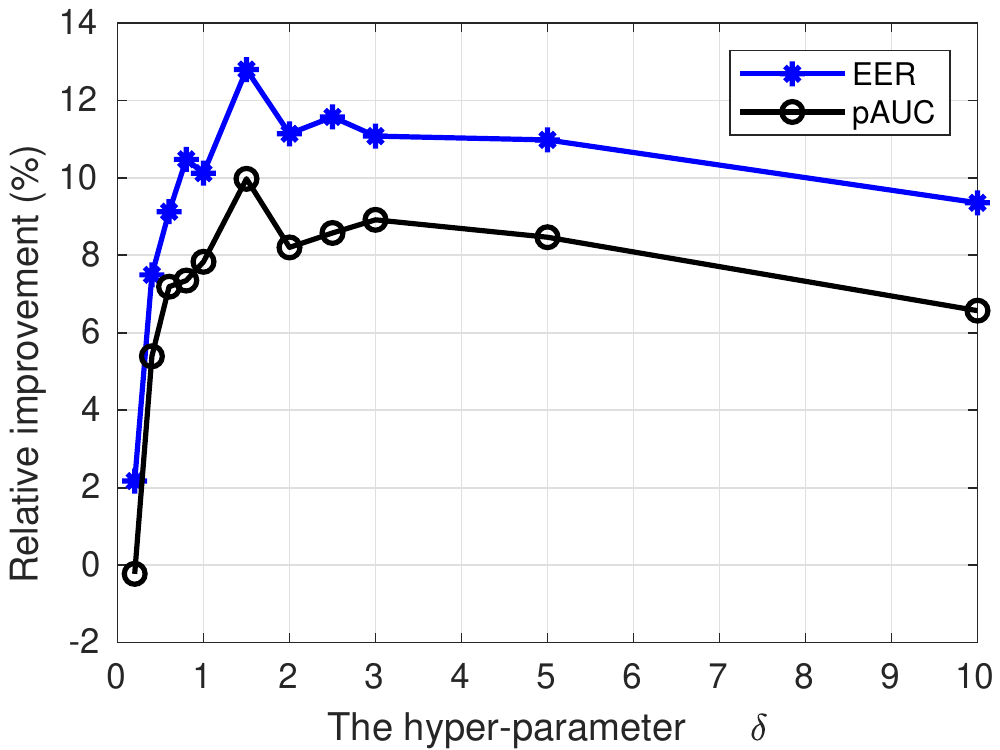}
  \caption{Relative performance improvement of pAUCMetric over PLDA with respect to hyperparameter $\delta$.}
  \label{fig:EER_delta}
\end{figure}

 We set the dimensions of the input features of the comparison back-ends from 50 to 400 with a step size of 50, where the features are produced from LDA. The experimental results are summarized in Table \ref{tab:dimension_analysis}.
From this table, one can see that pAUCMetric obtains lower EER scores and smaller performance variances than the comparison back-ends in all cases. It reaches the lowest EER when the input feature dimension is set to 150.

\subsubsection{Effects of the hyperparameters of pAUCMetric}\label{sec:parameter}
pAUCMetric has five hyperparameters $\alpha$, $\beta$, $\delta$, $\gamma$, and $\mu$.
The reason why we set $\alpha=0$ and $\beta=0.01$ is that the number of the imposter trials is much larger than the number of the true trials, hence restricting the working area $[\alpha, \beta]$ to a FPR range of close to zero makes the algorithm focus on discriminating the difficult trials.

We study the effects of $\delta$, $\gamma$, and $\mu$ by grid search. We first search $\delta$ in $[0,10]$ with the other hyperparameters set to their default values. Figure~\ref{fig:EER_delta} shows the relative performance improvement of pAUCMetric over PLDA. From the figure, we find that pAUCMetric is robust in a wide range of $\delta$ with the best $\delta$ being around 1.5. We search $\gamma$ and $\mu$ in grid jointly as listed in Tables \ref{tab:gamma_mu_EER} and \ref{tab:gamma_mu_pAUC} with the other hyperparameters set to their default values. It is observed that the stable working region is $\mu\in[0, 10^{-3}]\cap \gamma\in[0,1.5]$. Interestingly, pAUCMetric still works well even without regularization, i.e., $\mu = 0$ and $\gamma=0$. The above observation is consistent across all training scenarios of this paper. To summarize, pAUCMetric is insensitive to the 3 hyperparameters.

\subsubsection{Complexity analysis}\label{sec:Complexity}

\begin{figure}[t!]
  \centering
  \includegraphics[width=2.4in]{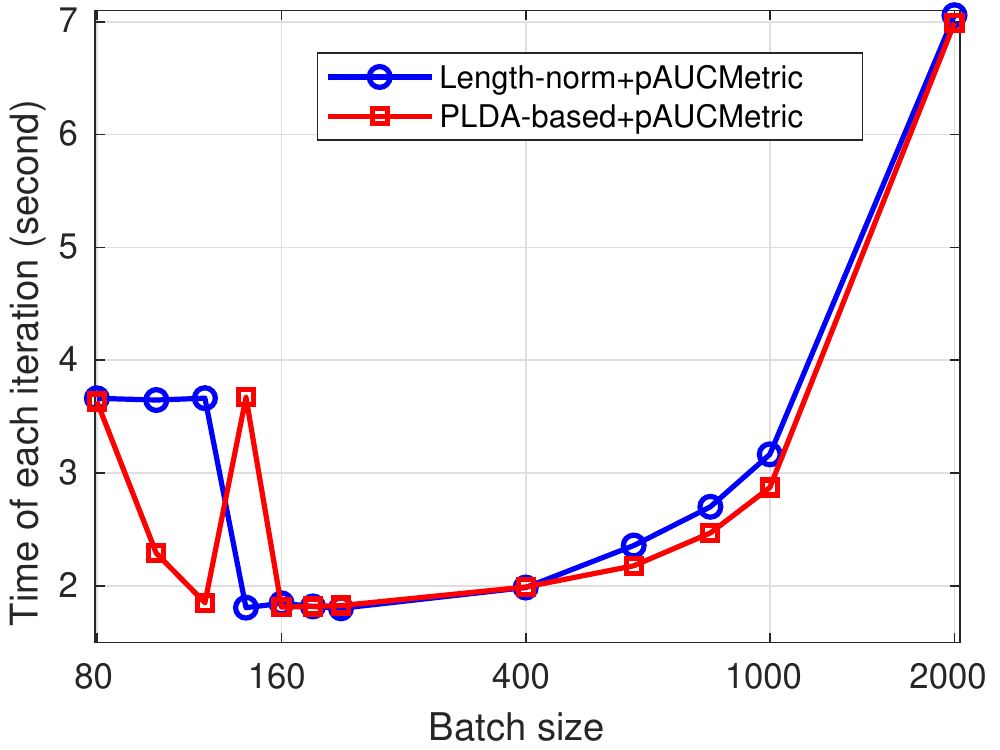}
  \caption{Training time of pAUCMetric at each iteration with different batch sizes.}
  \label{fig:Time_variable_batch_size}
\end{figure}

\label{sec:complexity}
\begin{figure}[t!]
  \centering
  \includegraphics[width=2.6in]{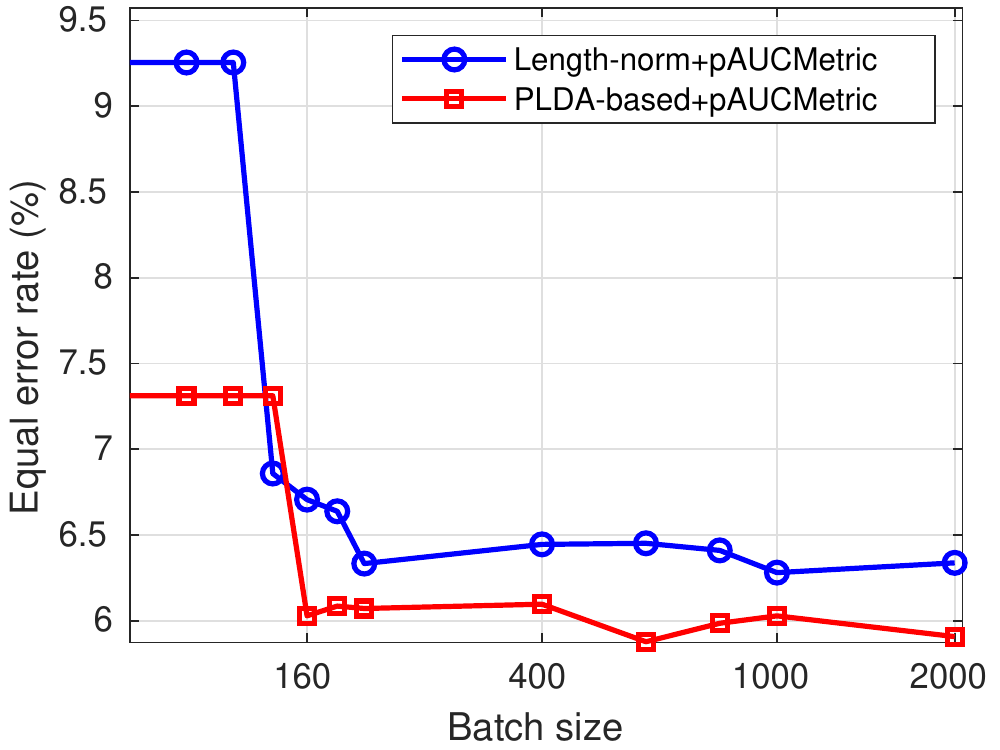}
  \caption{EER of pAUCMetric with different batch sizes.}
  \label{fig:EER_variable_batch_size}
\end{figure}

In Section \ref{sec:complexity}, we have  proven that the computational complexity is cubic with respect to the batch size $s$. This section further discusses the effect of $s$ on the computational complexity and performance of pAUCMetric.

\begin{figure}[t!]
  \centering
  \includegraphics[width=2.6in]{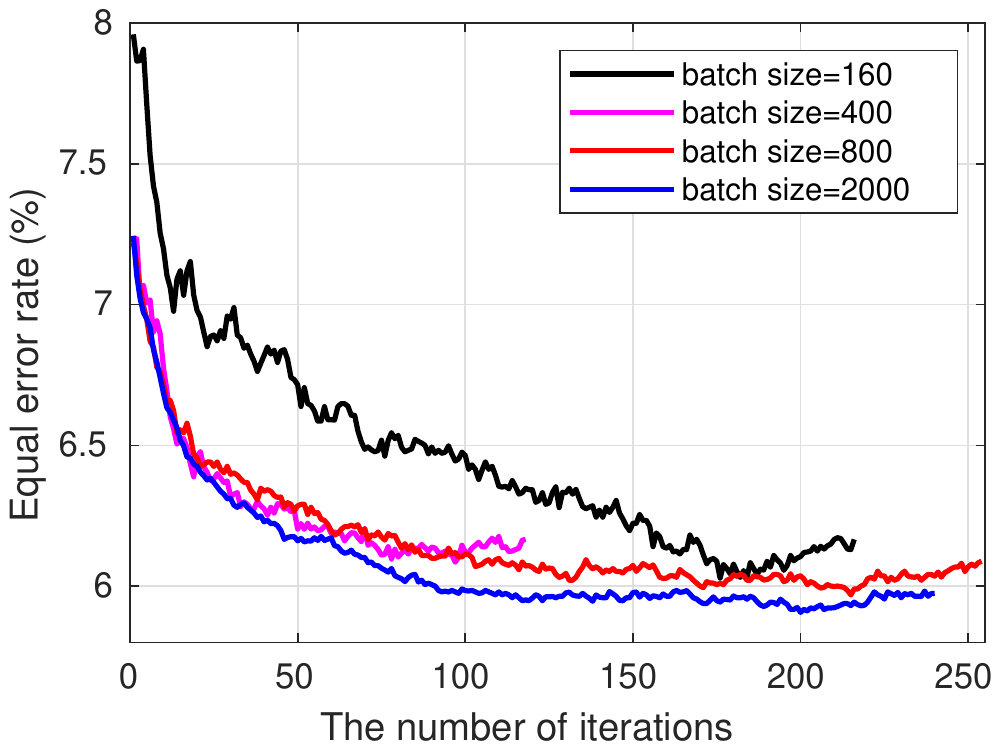}
  \caption{Convergence analysis of pAUCMetric with different batch sizes.}
  \label{fig:PLDA_pAUCMetric_EER_iterations}
\end{figure}

Figure \ref{fig:Time_variable_batch_size} shows the training time of \textit{the pAUCMetric at each iteration} with respect to the batch size $s$. One can see that the training time increases sharply with the value of $s$, which is consistent with the theoretical analysis in Section \ref{sec:complexity}. Note that, when the value of $s$ is less than 160, the fluctuation of the training time is caused by some random factors.

Figure \ref{fig:EER_variable_batch_size} shows the EER results of pAUCMetric with different values of $s$. One can see that, on the one hand, the value of $s$ cannot be too small, e.g. smaller than 160, and on the other hand, increasing the batch size does not always improve the performance. In practice, we only need a small suitable batch size, such as our default $s=500$.

Figure \ref{fig:PLDA_pAUCMetric_EER_iterations} plots the convergence rate with respect to $s$. We see that, when $s$ is larger than a reasonable small value, the convergence rate of pAUCMetric does not improve anymore. In other words, although the computational complexity of pAUCMetric is theoretically cubic with respect to $s$, setting $s$ to a small reasonable value not only guarantees good performance but also is efficient.

Finally, we evaluated the proposed pAUCMetric in other test scenarios beyond the scenarios in this subsection and with the length-normalization preprocessing technique as well. The experimental conclusions are consistent with those in this subsection. But we will not report the tedious results to make the paper concise.

\section{Conclusions and future work} \label{sec:summary}
In this paper, we presented a speaker verification back-end based on the squared Mahalanobis distance, i.e., pAUCMetric, to maximize pAUC. Because directly optimizing pAUC is an NP-hard problem, we first relaxed the optimization problem to a polynomial-time solvable one, and then adopted a random sampling strategy to reduce the computational complexity. The pAUC optimization was proven to be a problem of enlarging the weighted margin between the positive and negative trials, where the information of pAUC is encoded in the weights of the trials. In order to boost the performance of pAUCMetric, we further proposed to use the length-normalization and the PLDA-based preprocessing techniques. Experimental results on the  NIST 2016 SRE and SITW data demonstrated the effectiveness of pAUCMetric and showed that pAUCMetric is insensitive to the hyperparameter settings in all the studied evaluation scenarios.

The proposed method can be further improved in many aspects. Work is in progress to study automatic hyperparameter tuning algorithms via auto machine learning and investigate new methods that do not need feature preprocessing. Since pAUCMetric does not need a decision threshold, it is interesting to explore whether the pAUC optimization can be integrated with score calibration, which will be carried out in the near future. A speaker verification system consists of both front-end and  back-end. So, only developing a good back-end may not give the best performance. Consequently, it is legitimate to extend pAUCMetric to end-to-end training, which is also on our roadmap. Furthermore, as suggested by one anonymous reviewer, it is interesting to separate the effects of back-ends and loss functions, and evaluate how well the hinge loss can approximate the indicator function in pAUCMetric.

\section*{acknowledgements}

The authors are grateful to Dr. Kong Aik Lee, the Associate Editor, and the anonymous reviewers for their valuable comments, which helped greatly improve the quality of the paper.

\begin{appendices}
      \section{  }
A probabilistic explanation of the Mahalanobis distance is given as follows.
 Let $\mathbf{z}= \mathbf{x}_1-\mathbf{x}_2$ be the difference between two embedding vectors. We further assume that $p(\mathbf{z}|tar)=N(0,\mathbf{\Sigma}_0)$ and $p(\mathbf{z}|non)=N(0,\mathbf{\Sigma}_1)$, where ``$tar$'' and ``$non$'' denote target and non-target respectively.  The LLR test is:
 \begin{equation}\label{eq:LLR}
   LLR (\mathbf{z})= \log (p(\mathbf{z}|tar))-\log(p(\mathbf{z}|non)),
 \end{equation}
which can be transformed to:
 \begin{equation}\label{eq:LLR2}
   2LLR(\mathbf{z}) = -\mathbf{z}^T(\mathbf{\Sigma}_0^{-1}-\mathbf{\Sigma}_1^{-1})\mathbf{z} + \log|\mathbf{\Sigma}_1|-\log|\mathbf{\Sigma}_0|.
 \end{equation}
Neglecting the constant terms of \eqref{eq:LLR2} gives:
 \begin{equation}
    \widetilde{LLR}(\mathbf{z}) = -\mathbf{z}^T\mathbf{M}\mathbf{z}
 \end{equation}
where $\mathbf{M}=\mathbf{\Sigma}_0^{-1}-\mathbf{\Sigma}_1^{-1}$ is the parameters of the Mahalanobis distance \cite{koestinger2012large}.

    \section{  }
 \begin{myTheo}
   The relative constraints of the triplet-loss are a subset of the relative constraints of the AUC-loss.
 \end{myTheo}
 \begin{proof}
Let $\mathbf{x}_i^m$ be the $i$th embedding vector of the $m$th speaker.
The relative constraints of the triplet-loss $Tri$ is:
\begin{equation}
  Tri=\{(\mathbf{x}_i^m, \mathbf{x}_j^m; \mathbf{x}_k^n)|i \not= j,m \not= n\}
\end{equation}
The relative constraints of the AUC-loss $Tet$ can be divided into the following four sets:
\begin{equation}
  Tet_1=\{(\mathbf{x}_i^m, \mathbf{x}_j^m;\mathbf{x}_i^m,\mathbf{x}_k^n)|i \not= j,m \not= n\}
\end{equation}
\begin{equation}
  Tet_2=\{(\mathbf{x}_i^m, \mathbf{x}_j^m;\mathbf{x}_j^m,\mathbf{x}_k^n)|i \not= j,m \not= n\}
\end{equation}
\begin{equation}
  Tet_3=\{(\mathbf{x}_i^m, \mathbf{x}_j^m;\mathbf{x}_l^m,\mathbf{x}_k^n)|i \not= j \not= l,m \not= n\}
\end{equation}
\begin{equation}
  Tet_4=\{(\mathbf{x}_i^m, \mathbf{x}_j^m;\mathbf{x}_l^t,\mathbf{x}_k^n)|i \not= j,m \not= t \not= n\}
\end{equation}
with $Tet= Tet_1 \cup Tet_2 \cup Tet_3 \cup Tet_4 $.
Obviously, $Tet_1 \cup Tet_2 = Tri$, which derives $Tri \subset Tet$.
\end{proof}

\section{ }

The DET curves of the studied methods on the $\rm E2 \sim  E4$ schemes are plotted in Figs.~\ref{fig:11} to \ref{fig:13}.
\begin{figure}[h!]
  \centering
  \includegraphics[width=3.45in]{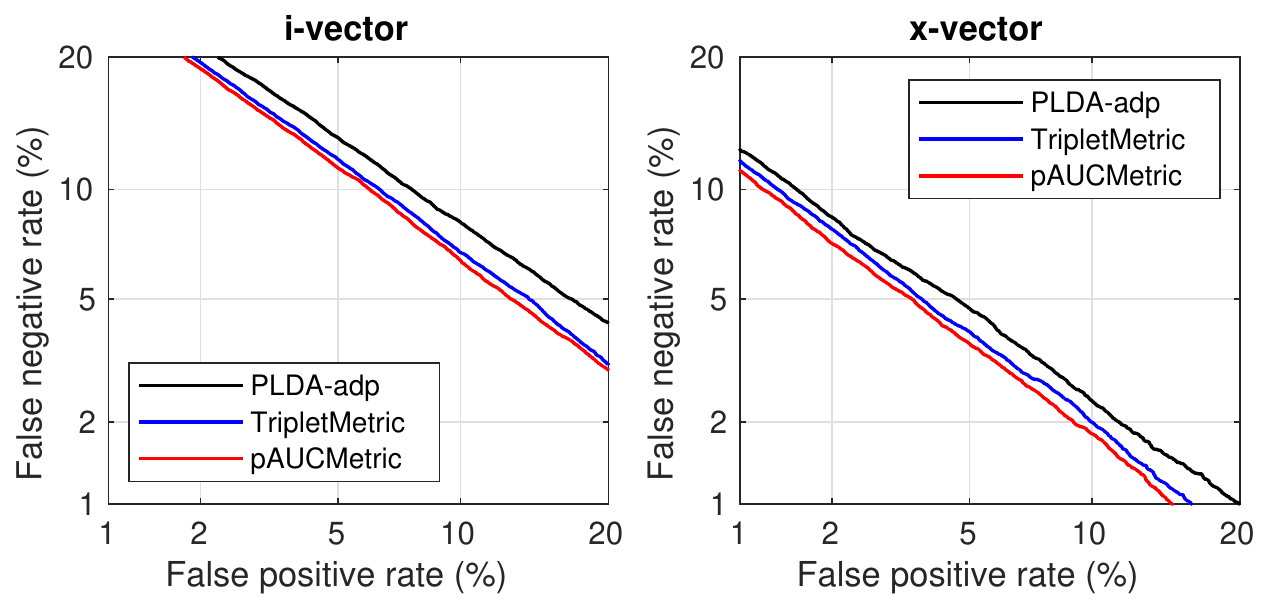}
  \caption{ DET curves of the studied  methods on the Cantonese data of the E2 scheme.}
  \label{fig:11}
\end{figure}

\begin{figure}[h!]
  \centering
  \includegraphics[width=3.45in]{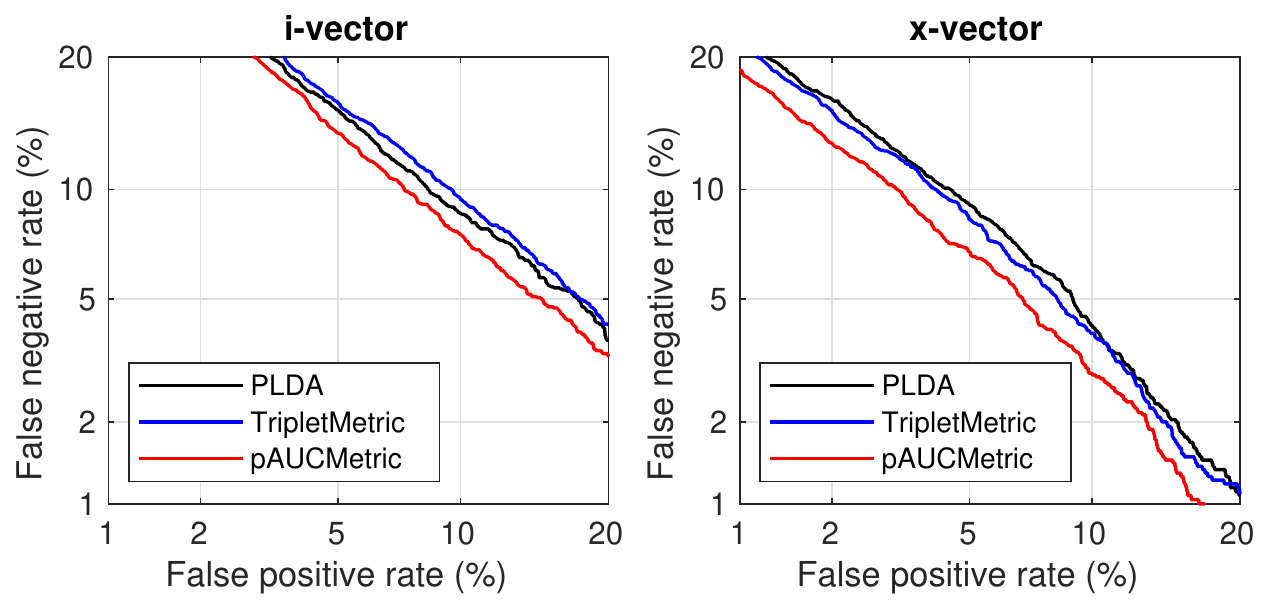}
  \caption{ DET curves of the studied methods on the Dev.Core task of the E3 scheme.}
\end{figure}

\begin{figure}[h!]
  \centering
  \includegraphics[width=3.45in]{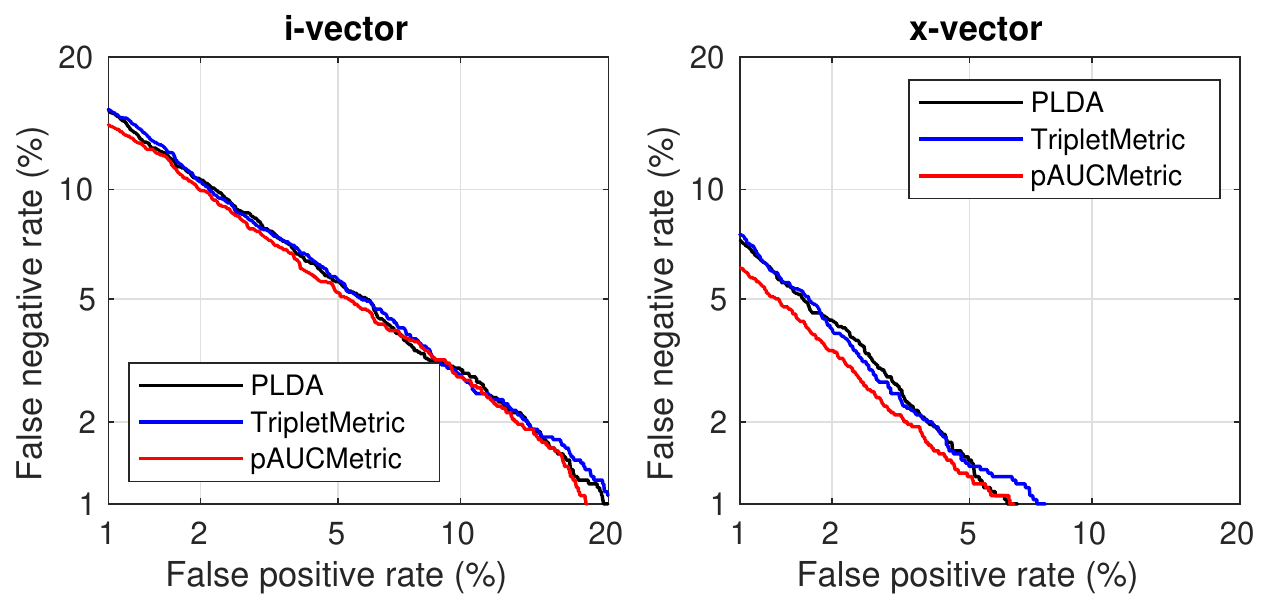}
  \caption{ DET curves of the studied  methods on the Dev.Core task of the E4 scheme.}
  \label{fig:13}
\end{figure}

\end{appendices}

\bibliographystyle{IEEEtran}
\bibliography{mybib}

\end{document}